\documentclass[10pt,numbers,preprint]{sigplanconf}

\usepackage{amsmath}
\usepackage{url}
\usepackage{amsthm,mathtools,amssymb,proof,stmaryrd}
\usepackage[linesnumbered]{algorithm2e}

\usepackage{listings}
\usepackage{tikz}
\usetikzlibrary{patterns,positioning,backgrounds}

\usepackage{prooftree}
\usepackage{enumitem}

\newtheorem{theorem}{Theorem}[section]
\newtheorem{proposition}[theorem]{Proposition}
\newtheorem{lemma}[theorem]{Lemma}

\theoremstyle{definition}
\newtheorem{definition}[theorem]{Definition}
\newtheorem{remark}[theorem]{Remark}
\newtheorem{example}[theorem]{Example}

\theoremstyle{plain}
\newcommand{\thistheoremname}{}
\newtheorem*{genericthm}{\thistheoremname}
\newenvironment{problem}[1]
  {\renewcommand{\thistheoremname}{#1}\begin{genericthm}}{\end{genericthm}}

\makeatletter
\newcommand{\dotminus}{\mathbin{\text{\@dotminus}}}

\newcommand{\@dotminus}{%
  \ooalign{\hidewidth\raise1ex\hbox{.}\hidewidth\cr$\m@th-$\cr}%
}
\makeatother

\renewcommand{\iff}{\Leftrightarrow}
\DeclarePairedDelimiter{\set}{\lbrace}{\rbrace}
\renewcommand{\vec}[1]{\mathbf{#1}}

\newcommand{\defeq}{=_{\textrm{\scriptsize{def}}}}

\newcommand{\nat}{\mathbb{N}}

\newcommand{\hash}{\mathrel{\#}}

\newcommand{\partialfn}{\rightharpoonup}

\newcommand{\finpartialfn}{\partialfn_{\textrm{\tiny fin}}}

\newcommand{\var}{\mathsf{Var}}
\newcommand{\nil}{\mathsf{nil}}

\newcommand{\emp}{\mathsf{emp}}
\newcommand{\psto}[2]{{#1} \mapsto {#2}}

\newcommand{\fv}[1]{FV(#1)}
\newcommand{\ASL}{\mathsf{ASL}}
\newcommand{\qf}[1]{\mathsf{qf}(#1)}
\newcommand{\trms}[1]{T(#1)}
\newcommand{\pres}{\mathsf{PbA}}
\newcommand{\abdterms}[2]{\mathcal{T}_{#1,#2}}
\newcommand{\htriple}[3]{\mbox{$\{#1\}\,#2\,\{#3\}$}}

\newcommand{\dom}[1]{\mathrm{dom}\left(#1\right)}

\newcommand{\val}{{\sf Val}}

\newcommand{\bigsepstar}{\mathop{\raisebox{-1ex}{{\Huge $*$}}}}
\def\bigsepcirc{\displaystyle\mathop{\raisebox{0ex}{$\bigcirc$}}}

\newcommand{\modelsPA}{\models}
\newcommand{\Dlt}{<_{\Delta}}
\newcommand{\Dleq}{\leq_{\Delta}}
\newcommand{\Deq}{=_{\Delta}}
\newcommand{\heaplet}[2]{\left\langle\psto{#1}{#2}\right\rangle}
\newcommand{\code}[1]{\text{\texttt{#1}}}

\newcommand{\arrcovalg}{\ensuremath{\mathsf{arrcov}}}
\newcommand{\arrcov}[3]{\ensuremath{\arrcovalg_{#1}(#2,#3)}}
\newcommand{\ptocovalg}{\ensuremath{\mathsf{ptocov}}}
\newcommand{\ptocov}[3]{\ensuremath{\ptocovalg_{#1}(#2,#3)}}

\newcommand{\mkarrays}[1]{\mathcal{#1}^{\mathsf{arr}}}
\newcommand{\mkptos}[1]{\mathcal{#1}^{\mathsf{pto}}}

\newcommand{\Aarrays}{\mkarrays{A}}
\newcommand{\Aptos}{\mkptos{A}}
\newcommand{\Xarrays}{\mkarrays{X}}
\newcommand{\Xptos}{\mkptos{X}}
\newcommand{\Barrays}{\mkarrays{B}}
\newcommand{\Bptos}{\mkptos{B}}
\newcommand{\Yarrays}{\mkarrays{Y}}
\newcommand{\Yptos}{\mkptos{Y}}
\newcommand{\HB}{\mathcal{H}_B}
\newcommand{\cutheap}[1]{\llbracket #1 \rrbracket^{s,h}}

\newcommand{\relarray}[3]{\mathsf{array}(#1,#2,#3)}
\newcommand{\absarray}[2]{\mathsf{array}(#1,#2)}
\newcommand{\ls}[2]{\mathsf{ls}\,#1\,#2}
\newcommand{\tool}[1]{\textsc{#1}}
\newcommand{\abstr}[1]{\lfloor #1\rfloor}

\newcommand{\NP}{\ensuremath{\mathsf{NP}}}
\newcommand{\CoNP}{\ensuremath{\mathsf{coNP}}}
\newcommand{\EXPTIME}{\ensuremath{\mathsf{EXP}}}
\newcommand{\EXP}{\EXPTIME}

\newcommand{\PTIME}{\ensuremath{\mathsf{PTIME}}}

\newcommand{\proofcase}[1]{\medskip\noindent{\bf Case {#1}:}}

\newcommand{\tikzarray}[6]{
\pgfmathsetmacro\sx{#3 - #1}
\pgfmathsetmacro\sy{#4 - #2}
\node (n#1#2#3#4) at (#1,#2) [draw,minimum width=\sx cm,minimum height=\sy cm,anchor=north west] {};
\node (l#1#2#3#4) at (#1,#2) [anchor=north west,text height=1.5ex] {#5};
\node (r#1#2#3#4) at (#3,#2) [anchor=north east,text height=1.5ex] {#6};
}

\newcommand{\tikzdashedarray}[6]{
\pgfmathsetmacro\sx{#3 - #1}
\pgfmathsetmacro\sy{#4 - #2}
\node (n#1#2#3#4) at (#1,#2) [draw,dashed,minimum width=\sx cm,minimum height=\sy cm,anchor=north west] {};
\node (l#1#2#3#4) at (#1,#2) [anchor=south west] {#5};
\node (r#1#2#3#4) at (#3,#2) [anchor=south east] {#6};
}

\newcommand{\tikzgap}[6]{
\pgfmathsetmacro\sx{#3 - #1}
\pgfmathsetmacro\sy{#4 - #2}
\node (n#1#2#3#4) at (#1,#2) [minimum width=\sx cm,minimum height=\sy cm,anchor=north west,pattern=north west lines] {};
\node (l#1#2#3#4) at (#1,#2) [anchor=south west] {#5};
\node (r#1#2#3#4) at (#3,#2) [anchor=south east] {#6};
}

\makeatletter
\newcommand{\mytag}[2]{%
  \text{#1}%
  \@bsphack
  \protected@write\@auxout{}%
         {\string\newlabel{#2}{{\theequation#1}{\thepage}}}%
  \@esphack
}
\makeatother

\makeatletter
\newcommand{\removelatexerror}{\let\@latex@error\@gobble}
\makeatother

\begin{document}

\title{Biabduction (and Related Problems) in Array Separation Logic}



 \authorinfo{James Brotherston}
            {University College London, UK}
            {J.Brotherston@ucl.ac.uk}

 \authorinfo{Nikos Gorogiannis}
            {Middlesex University, UK}
            {nikos.gorogiannis@gmail.com}

 \authorinfo{Max Kanovich}
            {University College London, UK
  and National Research University Higher School of Economics,
  Russian Federation}
            {M.Kanovich@ucl.ac.uk}

\maketitle

\begin{abstract}
We investigate \emph{array separation logic} ($\ASL$), a variant of symbolic-heap separation logic in which the data structures are either pointers or \emph{arrays}, i.e., contiguous blocks of allocated memory.  This logic provides a language for compositional memory safety proofs of 
imperative array programs.

We focus on the \emph{biabduction} problem for this logic, which has been established as the key to automatic specification inference at the industrial scale. 
We present an $\NP$ decision procedure for biabduction in $\ASL$ that produces solutions of reasonable quality, and we also show that the problem of finding a consistent solution is $\NP$-hard.

Along the way, we study satisfiability and entailment in our logic, giving decision procedures and complexity bounds for both problems. We show satisfiability to be $\NP$-complete, and entailment to be decidable with high complexity. The somewhat surprising fact that biabduction is much simpler than entailment is explained by the fact that, as we show, the element of choice over biabduction solutions enables us to dramatically reduce the search space.
\end{abstract}

\category{F.3.1}{Specifying and Verifying and Reasoning about Programs}{Logics of programs}
\category{D.2.4}{Software/Program Verification}{Assertion checkers}
\category{F.2}{Analysis of Algorithms and Problem Complexity}{}

\keywords Separation logic, arrays, biabduction, satisfiability, entailment, complexity.

\section{Introduction}
\label{sec:introduction}

In the last 15 years, \emph{separation logic}~\cite{Reynolds:02} has evolved from a novel way to reason about memory pointers to a mainstream technique for scalable program verification. Facebook's \tool{Infer}~\cite{Calcagno-etal:15} static analyser 
is perhaps the best known tool based on separation logic; other examples include \tool{SLAyer}~\cite{Berdine-Cook-Ishtiaq:11}, \tool{VeriFast}~\cite{Jacobs-etal:11} and the \tool{HiP} tool series~\cite{Chin-etal:12}.

Separation logic is based upon \emph{Hoare triples} of the form $\htriple{A}{C}{B}$, where $C$ is a program and $A,B$ are formulas in a logical language.  Its compositional nature, the key to scalable analysis, is supported by two main pillars.  The first pillar is the soundness of the following \emph{frame rule}:
\[\begin{prooftree}
\htriple{A}{C}{B}
\justifies
\htriple{A * F}{C}{B * F}
\using \mbox{(Frame})
\end{prooftree}\]
where the \emph{separating conjunction} $*$ is read, intuitively, as ``\emph{and separately in memory}'', and subject to the restriction that $C$ does not modify any free variables in 
$F$~\cite{Yang-OHearn:02}.

The second pillar is a tractable algorithm for the \emph{biabduction} problem~\cite{Calcagno-etal:11}: given formulas $A$ and $B$, find ``antiframe'' and ``frame'' formulas $X$, $Y$ respectively with
    \[
    A * X \models B * Y\ ,
    \]%
usually subject to the proviso that $A * X$ should be satisfiable.
Solving this problem enables us to infer specifications for whole 
programs given specifications for their individual
components~\cite{Calcagno-etal:11}.  E.g., if $C_1$
and $C_2$ have specifications $\htriple{A'}{C_1}{A}$ and
$\htriple{B}{C_2}{B'}$, we can use a solution $X,Y$ to the above
biabduction problem to construct a specification for $C_1;C_2$ as
follows, using the frame rule and the usual Hoare logic rules for
consequence ($\models$) and sequencing ($;$):  
\begingroup
\small 
\[\begin{prooftree}
\[
\[
\htriple{A'}{C_1}{A}
\justifies
\htriple{A' * X}{C_1}{A * X} \using \mbox{(Frame)}\]
\justifies
\htriple{A' * X}{C_1}{B * Y} \using \mbox{($\models$)}\]
\[
\htriple{B}{C_2}{B'}
\justifies
\htriple{B * Y}{C_2}{B' * Y} \using \mbox{(Frame)}\]
\justifies
\htriple{A' * X}{C_1;C_2}{B' * Y}
\using \mbox{($;$)}
\end{prooftree}\]
\endgroup
Bottom-up interprocedural analyses based on separation logic, such as Facebook \tool{Infer}, employ biabduction in this way to infer program specifications bottom-up from unannotated code. Typically, the underlying language of assertion formulas is based on the ``symbolic heap'' fragment of separation logic over linked lists~\cite{Berdine-Calcagno-OHearn:04}, which is known to be tractable~\cite{Cook-etal:11}. 

In this paper, we instead focus on a different, but similarly ubiquitous data structure for imperative programming, namely \emph{arrays}, which we view as contiguous blocks of allocated heap memory.  We propose an \emph{array separation logic} ($\ASL$) in which we replace the usual ``list segment'' predicate $\textsf{ls}$ of separation logic by an ``array'' predicate $\absarray{a}{b}$, which denotes a contiguous block of allocated heap memory from address $a$ to address $b$ (inclusive), as was first proposed in~\cite{OHearn:07}.  In addition, since we wish to reason about array bounds, we also allow our assertions to contain linear arithmetic. 
Thus, for example, a pointer $x$ to an memory block of length $n > 1$ and starting at $a$ can be represented in $\ASL$ by the assertion
\[
n > 1: x \mapsto a * \absarray{a}{a + n - 1}\ .
\]
The $\textsf{array}$ predicate only records the bounds of memory blocks, and not their contents; this is analogous to the abstraction from pointers to lists in standard separation logic. 
Indeed, 
the memory safety of array-manipulating programs typically depends only on the memory footprint of the arrays.  E.g., the usual {\tt quicksort} and {\tt mergesort} procedures for arrays work by partitioning the array and recursing based on a pivot chosen from among the allocated addresses.

Our focus in this paper is on the biabduction problem, as above, for $\ASL$. Solving this problem is, we believe, \emph{the} most critical step in building a bottom-up memory safety analysis  {\`a} la \tool{Infer} for array-manipulating programs. The first main contribution of the current work is a decision procedure for the (quantifier-free) biabduction problem in $\ASL$, which we present in Section~\ref{sec:biabduction}. It relies on the idea that, given $A$ and $B$, we can look for some consistent total ordering of all the array endpoints and pointer addresses in both $A$ and $B$, and impose this ordering, which we call a \emph{solution seed}, as the arithmetical part of the solution $X$.  Having done this, the computation of the ``missing'' arrays and pointers in $X$ and $Y$ becomes a polynomial-time process, and thus the entire algorithm runs in $\NP$-time. We demonstrate that, as well as being sound, this algorithm is in fact complete; a biabduction solution exists if and only if a solution seed exists.  We also show that the biabduction problem is $\NP$-hard, 
and give further bounds for cases involving quantifiers.

Along the way, we study the \emph{satisfiability} and \emph{entailment} problems in $\ASL$, and, as our second main contribution, we provide decision procedures and upper/lower complexity bounds for both problems. We find that satisfiability is $\NP$-complete, while entailment is decidable with very high complexity: it can be encoded in $\Pi^0_2$ Presburger arithmetic, and is also at least \mbox{$\Pi^P_2$-hard}.  
The fact that entailment is much harder than biabduction may at first sight appear surprising, since biabduction also seems to involve solving an entailment problem. 
However, in the biabduction problem, there is an element of \emph{choice} over $X$ and $Y$, and we can exploit this in a way that dramatically reduces the cost of checking these conditions. Namely, committing to a specific solution seed (see above) reduces biabduction to a simple computation rather than a search problem.

The remainder of this paper is structured as follows. Section~\ref{sec:examples} gives some examples showing how the $\ASL$ biabduction problem arises in verification practice; the syntax and semantics of $\ASL$ is then presented formally in Section~\ref{sec:language}.  We present algorithms and establish complexity bounds for satisfiability, biabduction and entailment for $\ASL$ in Sections~\ref{sec:satisfiability},~\ref{sec:biabduction} and~\ref{sec:entailment} respectively. Section~\ref{sec:related} surveys the related work, and Section~\ref{sec:conclusion} concludes.

Due to space limitations, most proofs of the results in this paper are either omitted or only sketched.  Full proofs are available in the supplementary material for referees.

\section{Motivating examples}
\label{sec:examples}

Here, we show two examples in order to illustrate how the biabduction problem arises in the context of verifying array programs, using $\ASL$ as the underlying assertion language.  In these examples, we often use a ternary \emph{base-offset} variant of the basic {\textsf{array}} predicate: $\relarray{a}{i}{j}$ is syntactic sugar for $\absarray{a+i}{a+j}$.  (See Remark~\ref{rem:addressing} for further details.)

We note that, in separation logic, Hoare triples $\htriple{A}{C}{B}$ have a \emph{fault-avoiding} interpretation, where the precondition $A$ guarantees that the code $C$ is memory safe~\cite{Yang-OHearn:02}.


\newcommand{\Cd}{\text{\lstinline|d|}}
\newcommand{\Cs}{\text{\lstinline|s|}}
\newcommand{\Cb}{\text{\lstinline|b|}}
\newcommand{\Cm}{\text{\lstinline|m|}}
\newcommand{\Cl}{\text{\lstinline|l|}}
\newcommand{\Ck}{\text{\lstinline|k|}}
\newcommand{\Cn}{\text{\lstinline|n|}}
\newcommand{\Cz}{\text{\lstinline|z|}}

\long\def\commentt#1{}

\begin{example}
\label{ex:biabduction}
A message \Cm{} of size \Ck{} must be inserted
at the beginning of a buffer \Cb{} of size \Cn{},
shifting the previous contents to the right, via
the following \lstinline|C| function \lstinline|shift_put|.

\begingroup
\footnotesize
\begin{lstlisting}[escapeinside={(*}{*)}]
void shift_put(char *m, int k, char *b, int n){
  memmove(b+k, b, n-k);     // (*$c_1$*)
  memcpy(b, m, k);          // (*$c_2$*)
}
\end{lstlisting}
\endgroup

The procedure\/ \lstinline|memmove(d,s,z)| copies a byte sequence of
length \Cz{} starting from address \Cs{} into the region starting at
address \Cd{}, even when the two regions overlap. The call $c_1$
should shift the previous contents of length $\Cn{}-\Ck{}$
from the beginning of the buffer to its end.
The relevant specification of the call $c_1$ is the Hoare triple
$\htriple{A}{c_1}{A}$ where (assuming non-negative \lstinline|int|s) $A$
is the following assertion in $\ASL$:
\[
\Ck<\Cn :
\relarray{\Cb}{0}{\Ck-1} * \relarray{\Cb}{\Ck}{\Cn-1}\ . 
\]


Similarly, \lstinline|memcpy(d,s,z)| copies \Cz{} bytes from \Cs{}  into \Cd{} but overlap is forbidden. The call $c_2$ above is specified by the triple $\htriple{B}{c_2}{B}$, where $B$ is the $\ASL$ formula
\[
\relarray{\Cm}{0}{\Ck-1} * \relarray{\Cb}{0}{\Ck-1}\ . 
\]%
It is easy to see that $A\not\models B$, so we cannot immediately combine these specifications to prove
 \lstinline|shift_put|. To overcome this, we solve the \emph{biabduction} problem: find formulas $X,Y$ such that \mbox{$A * X \models B * Y$}. One possible solution is
\[
 X = \relarray{\Cm}{0}{\Ck-1}\ , \quad Y = \relarray{\Cb}{\Ck}{\Cn-1}\ .
\]%
Using $A*X$ as the precondition of\/ \mbox{$c_1;c_2$}\ and the knowledge that \mbox{$A*X\models B*Y$} allows us to apply
 the derivation given earlier and\/ \emph{automatically abduce} the valid specification
$$
  \htriple{D}{\text{\lstinline|shift_put(m,k,b,n)|}}{D}
$$%
where (after merging arrays from \Cb{}) $D$ is the $\ASL$ assertion
\[
\Ck<\Cn : \relarray{\Cm}{0}{\Ck-1} * \relarray{\Cb}{0}{\Cn-1}\ . 
\]%
\end{example}

\long\def\commentt#1{}

\begin{example}
\label{ex:heapsort}
Here we show how to assemble a valid memory specification for
{\tt BUILD-MAX-HEAP}, an essential preparatory step
 in \emph{heapsort} (see e.g.~\cite{Cormen-etal:09}),
 starting from ``small specs'' for its atomic commands.

\begingroup
\footnotesize
\begin{lstlisting}[escapeinside={(*}{*)}]
void BUILD-MAX-HEAP(int n){
  int i = 1; T[i] = (*$b_{1}$*);     // root
  while (2*i <= n) {
    T[2*i] = (*$b_{2\code{i}}$*);           // left-hand child
    H(2*i);                // restore max-heap
    if (2*i+1 <= n) {
      T[2*i+1] = (*$b_{2\code{i}+1}$*);     // right-hand child
      H(2*i+1);            // restore max-heap
    }
    i = i+1;
  }
}
\end{lstlisting}
\endgroup

\noindent
Given a list of values $b_1,\ldots,b_n$, we insert them one-by-one into the array \code{T}, which is viewed
 as a \emph{binary tree}: for any index $i$, the indices $2i$ and $2i+1$
 are its left and right `children'.  Additionally, the values stored in the array should satisfy the \emph{max-heap property}: for every child $j$ with parent~$i$ say (so \mbox{$i = \lfloor j/2 \rfloor$}), we have \code{T[$j$]} $\leq$ \code{T[$i$]}. This property is maintained using the auxiliary function
 \code{H(k)}, which
 swaps up the newly added \code{T[k]} to the proper place along the path from \code{k} to the root, $1$:

\begingroup
\footnotesize
\begin{lstlisting}[escapeinside={(*}{*)}]
void H(int k){
 int j2 = k; int j1 = (*$\lfloor$*)j2/2(*$\rfloor$*);
 while ( j1 >= 1 && T[j1] < T[j2] ){
  swap(T[j1],T[j2]); j2 = j1; j1 = (*$\lfloor$*)j2/2(*$\rfloor$*);}}
\end{lstlisting}
\endgroup

\noindent
Fixing a base-offset $a$ for the array \code{T}, observe first that any command of the form \code{T[$j$] = $b_j$;}
obeys the obvious memory spec:
\[
\htriple
{\relarray{a}{j}{j}}
{\ \code{T[$j$] = $b_j$;}\ }{\psto{a + j}{b_j}}
\label{eq-j-heapsort}
\]
Now, writing $C_j$ for a command of the form \code{T[$j$] = $b_j$;}, observe that we can compose a specification for \mbox{$C_1;C_2;$}, as in the previous example, by solving the following biabduction problem: find $X$ and $Y$ such that
\[
 a + 1 \mapsto b_{1}\ *\ X \models
 \relarray{a}{2}{2} * Y\ .
     \label{eq-X-Y-heapsort}
\]
A minimal solution is quite evident: take $X=\relarray{a}{2}{2}$ and $Y= \relarray{a}{1}{1}$. Following the method outlined in the introduction, we can (automatically) generate the valid specification:
\[%
\htriple
{\relarray{a}{1}{1} * X}{\ C_{1};C_{2};\ }{a + 2 \mapsto b_{2} * Y}
\label{eq-C1-C2-heapsort-0}
\]
Taking into account that 
$a + 2\mapsto b_2 \models \relarray{a}{2}{2}$, and joining arrays, we get
\[
 \htriple{\relarray{a}{1}{2}}{\ C_{1};C_{2};\ }
            {\relarray{a}{1}{2}}
     \label{eq-C1-C2-heapsort}
\]%
As \code{H($2$)} manipulates \code{T[$1$]} and \code{T[$2$]} only,
 we can show that
\[%
 \htriple{\relarray{a}{1}{2}}{\ C_{1};C_{2};\code{H($2$)};\ }
            {\relarray{a}{1}{2}}
     \label{eq-C1-C2-H-heapsort}
\]
By iterating this process, 
 we get a valid specification
 for the unfolded \code{BUILD-MAX-HEAP(n)}:
{
\[\begin{array}{c}
\{\relarray{a}{1}{\code{n}}\} \\[0.3ex]
{C_{1};C_{2};\code{H($2$)};C_{3};\code{H($3$)};\dots C_{\code{n}};\code{H(n)};} \\[0.3ex]
\{\relarray{a}{1}{\code{n}}\}
\end{array}\]}%
which is a valid specification for \code{BUILD-MAX-HEAP(n)}:
\[%
\htriple{\relarray{a}{1}{\code{n}}}
    {\ \code{BUILD-MAX-HEAP(n)}\ }{\relarray{a}{1}{\code{n}}}
\label{eq-heapsort}
\]%
\end{example}


\section{Array separation logic, $\ASL$}
\label{sec:language}

In this section we present separation logic for arrays, $\ASL$, which
employs a similar \emph{symbolic heap} formula structure to that
in~\cite{Berdine-Calcagno-OHearn:04}, but which treats contiguous
\emph{arrays} in memory rather than linked list segments; we
additionally allow a limited amount of pointer arithmetic, given by a conjunction of atomic Presburger formulas.


\begin{definition}[Symbolic heap]
\label{defn:sym_heap} \emph{Terms} $t$, \emph{pure formulas} $\Pi$ and \emph{spatial formulas} $F$ are given by the following grammar:
\begin{align*}
t & \Coloneqq x \mid n \mid t + t \mid nt \\
\Pi & \Coloneqq t = t \mid t \neq t \mid t \leq t \mid t < t \mid \Pi \land \Pi \\
F & \Coloneqq \emp \mid \psto{t}{t} \mid \absarray{t}{t} \mid F * F \\ 
SH & \Coloneqq \exists \vec{z}.\ \Pi : F
\end{align*}
where $x$ ranges over an infinite set $\var$ of \emph{variables},
$\vec{z}$ over sets of variables, and
$n$ over natural number constants in $\nat$. A \emph{symbolic heap} is
given by \mbox{$\exists \vec{z}.\ \Pi : F$}, where $\vec{z}$ is a tuple of
(distinct) variables, $F$ is a spatial formula and $\Pi$ is a pure formula.
Whenever one of $\Pi, F$ is empty, we omit the colon.
We write $\fv{A}$ for the set of free variables occurring in a symbolic
heap $A$.
\end{definition}

 If $A = \exists \vec{z}.\ \Pi : F$ is a symbolic heap,
 then we write $\qf{A}$ for $\Pi : F$,
 the \emph{quantifier-free} part of $A$.

We interpret the above language in a simple stack-and-heap model, in which we take both locations and values to be natural numbers. A \emph{stack} is a function $s\colon\var \rightarrow \nat$. We extend stacks to interpret terms in the obvious way:
\[
s(n) = n,\;  s(t_1 + t_2) = s(t_1) + s(t_2),\; \mbox{ and } s(nt) = ns(t)\ .
\]
If $s$ is a stack, $z \in \var$ and $m \in \nat$, we write $s[z \mapsto v]$ for the stack defined as $s$ except that $s[z \mapsto v](z) = v$. We extend stacks pointwise to act on tuples of terms.

A \emph{heap} is a finite partial function $h\colon \nat \finpartialfn \nat$ mapping finitely many locations to values; we write $\dom{h}$ for the set of locations on which $h$ is defined, and $e$ for the empty heap that is undefined on all locations.
We write $\circ$ for \emph{composition} of domain-disjoint heaps: if $h_1$ and $h_2$ are heaps, then $h_1 \circ h_2$ is the union of $h_1$ and $h_2$ when $\dom{h_1}$ and $\dom{h_2}$ are disjoint, and undefined otherwise.

\begin{definition}
\label{defn:satisfaction}
The \emph{satisfaction relation} $s,h\models A$, where $s$ is a stack, $h$ a heap and $A$ a symbolic heap, is defined by structural induction on $A$ in Fig.~\ref{fig:satisfaction}.
\end{definition}

\begin{figure}
\[\begin{array}{@{}l@{\hspace{0.1cm}}c@{\hspace{0.2cm}}l@{}}
s,h \models t_1 \sim t_2 & \Leftrightarrow & s(t_1) \sim s(t_2) \;\; (\mathord{\sim} \in \{=,\neq,<,\leq\}) \\[0.5ex]
s,h \models \Pi_1 \land \Pi_2 & \Leftrightarrow &
    s,h \models \Pi_1 \text{ and } s,h\models \Pi_2 \\[0.5ex]
s,h \models \emp & \Leftrightarrow & h = e \\[0.5ex]
s,h \models \psto{t_1}{t_2} & \Leftrightarrow & \dom{h}= \set{s(t_1)}\mbox{ and }\\&& h(s(t_1)) = s(t_2) \\[0.5ex]
s,h \models \absarray{t_1}{t_2} & \iff & s(t_1) \leq s(t_2) \mbox{ and } \\ && \dom{h}=\{s(t_1), \ldots, s(t_2)\}\\[0.5ex]
s,h \models F_1 * F_2 & \Leftrightarrow &
    \exists h_1,h_2.\ h=h_1 \circ h_2 \mbox{ and } \\&& s,h_1 \models F_1 \mbox{ and } s,h_2 \models F_2 \\[0.5ex]
s,h \models \exists \vec{z}.\ \Pi : F & \Leftrightarrow& \exists\vec{m} \in \val^{|\vec{z}|}.\ s[\vec{z} \mapsto \vec{m}],h \models \Pi \\&& \mbox{ and } s[\vec{z} \mapsto \vec{m}],h \models F
\end{array}\]%
\caption{The $\ASL$ satisfaction relation (cf.~Defn~\ref{defn:satisfaction}).}
\label{fig:satisfaction}
\end{figure}

Satisfaction of pure formulas $\Pi$ does not depend on the heap; we write $s \models \Pi$ to mean that $s,h \models \Pi$ (for any heap $h$). We write $A \models B$ to mean that $A$ \emph{entails} $B$, i.e. that $s,h \models A$ implies $s,h \models B$ for all stacks $s$ and heaps $h$.

\begin{lemma}
\label{lem:semantics}
For all quantifier-free symbolic heaps $A$, if \mbox{$s,h \models A$} and $s,h' \models A$, then $\dom{h} = \dom{h'}$.
\end{lemma}


\begin{remark}
\label{rem:addressing}
Our \textsf{array} predicate employs \emph{absolute} addressing:
$\absarray{k}{\ell}$ denotes an array from $k$ to $\ell$. In practice,
one often reasons about arrays using \emph{base-offset} addressing,
where $\relarray{b}{i}{j}$ denotes an array from $b+i$ to $b+j$. We can
define such a ternary version of our \textsf{array} predicate,
overloading notation, by:
\[
\relarray{b}{i}{j} \;\defeq\; \absarray{b+i}{b+j}
\]%
Conversely, any $\absarray{k}{\ell}$ can be represented in base-offset style as $\relarray{0}{k}{\ell}$. The moral is that we may freely switch between absolute and base-offset addressing as desired.
\end{remark}

In order to obtain sharper complexity results, we will sometimes
 confine our attention to symbolic heaps in the following special\/
\emph{two-variable form}.

\begin{definition}   
\label{defn:twovar}
 A symbolic heap \mbox{$\exists\vec{z}.\ \Pi\colon F$} is said
 to be in {\em two-variable form\/} if
\begin{itemize}
\item[(a)] its pure part $\Pi$ is a conjunction of
 {\em `difference constraints'} of the form
 \mbox{$x=k$}, \mbox{$x=y+k$}, \mbox{$x\leq y+k$},
  \mbox{$x\geq y+k$}, \mbox{$x < y+k$}, and \mbox{$x > y+k$},
 where $x$ and $y$ are variables, and \mbox{$k\in\nat$};
 (notice that \mbox{$x\neq y$} is not here);

\item[(b)] its spatial part $F$ contains only formulas  
 of the form\ \mbox{$k\mapsto v$},
 \mbox{$\relarray{a}{0}{j}$},\ \mbox{$\relarray{a}{1}{j}$},
 and\/ \mbox{$\relarray{k}{j}{j}$},
 where $v$, $a$, and\/ $j$ are variables, and \mbox{$k\in\nat$}.
\end{itemize}
\end{definition}

\begin{remark}\label{r-diff}
 The unrestricted pure part of our language is already
 \mbox{$\NP$-}hard. However, when we restrict pure formulas to
 conjunctions of `difference constraints'
 \footnote
  {The first order theory of such
   constraints is sometimes called ``difference logic'' or, amusingly
   enough, ``separation logic''~\cite{Talupur-etal:04}!}
 as in Definition~\ref{defn:twovar}, their satisfiability can be decided in
 polynomial time~\cite{Cormen-etal:09}. Therefore, restricting our
 symbolic heaps to two-variable form readdresses the challenge of
 establishing relevant lower bounds to the spatial part of the
 language.
\end{remark}

\section{Satisfiability in $\ASL$}
\label{sec:satisfiability}

Here, we show that \emph{satisfiability} in $\ASL$ is $\NP$-complete.
This stands in contrast to the situation for symbolic-heaps over list segments,
where satisfiability is polynomial~\cite{Cook-etal:11},
and over general inductive predicates, where it is $\EXPTIME$-complete \cite{Brotherston-etal:14}. The problem is stated formally as follows:

\begin{problem}{Satisfiability problem for $\ASL$}
Given symbolic heap $A$, decide whether there is a stack $s$ and heap
$h$ with \mbox{$s,h \models A$}.  (W.l.o.g., we may consider $A$ to be
quantifier-free.)
\end{problem}

First, we show that satisfiability of a symbolic heap can be encoded as
a $\Sigma^0_1$ formula of \emph{Presburger arithmetic}.

\begin{definition}
\label{defn:presburger}
\emph{Presburger arithmetic} ($\pres$) is defined as the first-order theory of the natural numbers $\nat$ over the signature $\langle 0, s, + \rangle$, where $s$ is the successor function, and $0$ and $+$ have their usual interpretations. It is immediate that the relations $\neq$, $\leq$ and $<$ can be encoded (possibly introducing an existential quantifier), as can the operation of multiplication by a constant.
\end{definition}

Note that a stack is just a standard first-order valuation, and that any pure formula in $\ASL$ is also a formula of $\pres$.  Moreover, the satisfaction relations for $\ASL$ and $\pres$ coincide on such formulas.
Thus, we overload $\models$ to include the standard first-order satisfaction relation of $\pres$.

The intuition behind our encoding of $\ASL$ satisfiability in $\pres$ is simple: a symbolic heap is satisfiable
exactly when the pure part is satisfiable, each array is well-defined, and all pointers and arrays are non-overlapping with all of the others. For simplicity of exposition, we do this by abstracting away pointers with single-cell arrays.

\begin{definition}
\label{defn:abstraction}
Let $A$ be a quantifier-free symbolic heap, written (without loss of generality) in the form:
\[
\Pi : \textstyle\bigsepstar^n_{i=1}\absarray{a_i}{b_i} * \bigsepstar^m_{i=1}\psto{c_i}{d_i}\ . \\
\]%
We define its \emph{array abstraction} $\abstr{A}$ as
\[
\Pi : \textstyle\bigsepstar^n_{i=1}\absarray{a_i}{b_i} * \bigsepstar^m_{i=1}\absarray{c_i}{c_i}\ .
\]
\end{definition}

\begin{lemma}
\label{lem:abstraction}
Let $A$ be a quantifier-free symbolic heap and $s$ a stack.
Then\; $\exists h.\ s,h\models A \;\iff\; \exists h'.\ s,h'\models \abstr{A}$.

\end{lemma}

\begin{definition}
\label{defn:gamma}
Let $A$ be a quantifier-free symbolic heap, and let $\abstr{A}$
be of the form $\Pi : \bigsepstar^n_{i=1}\absarray{a_i}{b_i}$.
We define a corresponding formula $\gamma(A)$ of $\pres$ as
\[
\gamma(A) \defeq \Pi \wedge \!\!\bigwedge_{1 \leq i \leq n}a_i \leq b_i \wedge
\!\!\bigwedge_{1 \leq i < j \leq n} \!\!\!(b_i < a_j) \vee (b_j < a_i)\ .
\]%
\end{definition}
\noindent Note that $\gamma(A)$ is defined in terms of the abstraction $\abstr{A}$.

\begin{lemma}
\label{lem:gamma}
For any stack $s$ and any quantifier-free 
symbolic heap $A$, we have $s \modelsPA \gamma(A)\;\iff\;  \exists h.\ s,h \models A$.
\end{lemma}

\begin{proposition}
\label{prop:sat_NP}
Satisfiability for $\ASL$ is in $\NP$.
\end{proposition}

\begin{proof}
Letting $\vec{x}$ be a tuple of all free variables
of a symbolic heap $A$, the $\Sigma^0_1$ Presburger arithmetic sentence
$\exists\vec{x}.\gamma(A)$, where $\gamma$ is given by
Definition~\ref{defn:gamma}, is of size quadratic
in the size of $A$. By Lemma~\ref{lem:gamma}, 
$A$ is satisfiable iff $\exists \vec{x}.\ \gamma(A)$ is
satisfiable. Since the satisfiability problem for $\Sigma^0_1$
Presburger arithmetic is in $\NP$~\cite{Scarpellini:84}, so is 
satisfiability for $\ASL$. 
\end{proof}

\begin{remark}
\label{r-sat-2-n}
Symbolic-heap separation logic on list segments~\cite{Berdine-Calcagno-OHearn:04} enjoys the {\em small model property}: any satisfiable formula~$A$ has a model of size polynomial in the size of\/~$A$~\cite{Antonopoulos-etal:14}.  Unfortunately, this property fails for~$\ASL$. 
E.g., let $A_n$ be a symbolic heap of the form
$$ 
  (d_0=1)\wedge \textstyle\bigwedge_{i=0}^{n-1} (d_{i} < d_{i+1})\, \colon\,
 \bigsepstar_{i=0}^{n} \relarray{d_i}{0}{d_i}\ .
$$%
Then we have that \mbox{$\bigwedge_{i=0}^{n-1} (s(d_{i+1}) > 2s(d_{i}))$} for any model $(s,h)$ of $A_n$, which implies that \mbox{$s(d_{n}) > 2^n$}, and so $h$ occupies a contiguous memory block of at least \/ $2^n$ cells.
\end{remark}

We establish that satisfiability is in fact {\em $\NP$-hard\/} by reduction from the {\em \mbox{$3$}-partition problem}~\cite{Garey-Johnson:79}.

\begin{problem}{3-partition problem~\cite{Garey-Johnson:79}}
Given a bound \mbox{$B\in\nat$} and a sequence of natural numbers
\mbox{${\cal S} = (k_1,k_2,\dots,k_{3m})$} such that
\mbox{$\sum_{j=1}^{3m}k_j= mB$},
 and, in addition,
 \mbox{$B/4<k_j< B/2$}\/\ \ for all \mbox{$1 \leq j \leq 3m$},
 decide whether there is a 
a partition of the elements of $\cal{S}$ into $m$ groups
 of three, say
\[\{(k_{j_{i,1}}, k_{j_{i,2}}, k_{j_{i,3}}) \mid 1 \leq i \leq m\}, \]%
 such that\/\ \mbox{$k_{j_{i,1}}+k_{j_{i,2}}+k_{j_{i,3}} = B$}\ \
 for all\/  \mbox{$1 \leq i \leq m$}.
\end{problem}

\begin{definition}
\label{defn:3part_to_sat}
Given an instance $(B,\cal{S})$ of the 3-partition problem,
we define a corresponding symbolic heap $A_{B,{\cal S}}$.

For convenience, we use the ternary ``base-offset'' version of our arrays to define $A_{B,\cal{S}}$, as given by Remark~\ref{rem:addressing}. First we introduce \mbox{$(m+1)$} variables $d_i$ and \mbox{$3m$} variables $a_j$. The idea is that the $d_i$ act as single-cell delimiters between chunks of memory of length~$B$, while the $a_j$ serve to allocate arrays of length $k_j$ in the space between some pair of delimiters $d_i$ and $d_{i+1}$.  The arrangement is as follows:
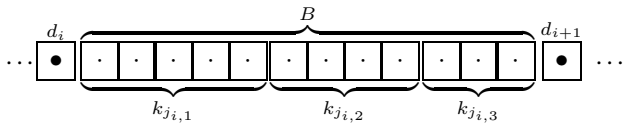
\newcount \DBL
 \DBL=13\multiply\DBL by 18%
$$\begin{picture}(\DBL,25)%
  \ \raisebox{6pt}{\ldots}%
  $\stackrel{d_i}{\framebox(14,14)[c]{$\bullet$}}$
$\overbrace{%
\underbrace{%
  \framebox(14,14)[c]{$\cdot$}%
  \framebox(14,14)[c]{$\cdot$}%
  \framebox(14,14)[c]{$\cdot$}%
  \framebox(14,14)[c]{$\cdot$}%
  \framebox(14,14)[c]{$\cdot$}%
}_{k_{j_{i,1}}}
\underbrace{%
  \framebox(14,14)[c]{$\cdot$}%
  \framebox(14,14)[c]{$\cdot$}%
  \framebox(14,14)[c]{$\cdot$}%
  \framebox(14,14)[c]{$\cdot$}%
}_{k_{j_{i,2}}}
\underbrace{%
  \framebox(14,14)[c]{$\cdot$}%
  \framebox(14,14)[c]{$\cdot$}%
  \framebox(14,14)[c]{$\cdot$}%
}_{k_{j_{i,3}}}
}^{B}$
  $\stackrel{d_{i+1}}{\framebox(14,14)[c]{$\bullet$}}$
  \ \raisebox{6pt}{\ldots}%
\end{picture}$$%
\vspace*{0.5ex}

\noindent
Concretely, $A_{B,\cal{S}}$ is the following symbolic heap:
\[\begin{array}{l}
\textstyle\bigwedge_{i=1}^{m} (d_{i+1}= d_{i} + B + 1)\; \wedge \\[1ex]
\bigwedge_{j=1}^{3m} (d_{1}\leq a_j)\wedge(a_j+k_j< d_{m+1}) : \\[1ex]
\bigsepstar_{i=1}^{m+1} \relarray{d_i}{0}{0} * \bigsepstar_{j=1}^{3m} \relarray{a_j}{1}{k_j}\ .
\end{array}\]%
where the indexed ``big star'' notation abbreviates a sequence of \mbox{$*$-conjoined} formulas.
We observe that $A_{B,\cal{S}}$ is quantifier-free and in two-variable
form (cf. Defn.~\ref{defn:twovar}).
\end{definition}

\begin{lemma}
\label{lem:3part_to_sat}
Given a 3-partition problem instance $(B,\cal{S})$, and letting $A_{B,\cal{S}}$ be the symbolic heap given by Defn.~\ref{defn:3part_to_sat}, 
\[
A_{B,\cal{S}} \mbox{ is satisfiable} \iff
 \exists\, \mbox{complete 3-partition of } \cal{S} \mbox{ (w.r.t. $B$).}
\]
\end{lemma}

\begin{theorem}
\label{thm:sat-NP-hard}
The satisfiability problem for $\ASL$ is $\NP$-complete, even for quantifier-free and \emph{$\mapsto$-free} symbolic heaps
 in two-variable form.
\end{theorem}

\begin{proof}
Proposition~\ref{prop:sat_NP} provides the upper bound.  For the lower bound, Defn.~\ref{defn:3part_to_sat} and Lemma~\ref{lem:3part_to_sat} establish a polynomial reduction from the \mbox{$3$-}partition problem.
\end{proof}

%


\section{Biabduction}
\label{sec:biabduction}

In this section, we turn to the central focus of this paper,
\emph{biabduction} for $\ASL$. 
In stating this problem, it is convenient to first lift the connective $*$ to symbolic heaps, as follows:
\[
(\exists \vec{x}.\ \Pi : F) * (\exists \vec{y}.\ \Pi' : F') = \exists \vec{x} \cup \vec{y}.\ \Pi \land \Pi' : F * F'\ ,
\]%
where we assume that the existentially quantified variables $\vec{x}$ and $\vec{y}$
are disjoint, and that no free variable capture occurs (this can always be avoided by $\alpha$-renaming).

\begin{problem}{Biabduction problem for $\ASL$}
 Given satisfiable symbolic heaps $A$ and\/~$B$,
 find symbolic heaps $X$ and\/~$Y$ such that
 \mbox{$A*X$} is satisfiable and \mbox{$A*X\models B*Y$}.
\end{problem}

We first consider quantifier-free biabduction,
i.e., where all of $A,B,X,Y$ are quantifier-free (\MakeUppercase sec.\nobreakspace \ref {sec:qfree-biabduction}).
The complexity of quantifier-free biabduction is investigated in \MakeUppercase sec.\nobreakspace \ref {sec:qfree-biabduction-bounds}.
We then show that when quantifiers appear in $B,Y$ which are appropriately restricted,
existence of solutions can be decided using the machinery we provide for the
quantifier-free case (\MakeUppercase sec.\nobreakspace \ref {sec:quantified-biabduction}).
In the same section we also characterise the complexity of biabduction in the
presence of quantifiers.

\subsection{An algorithm for quantifier-free biabduction}
\label{sec:qfree-biabduction}
We now present an algorithm for quantifier-free
biabduction.
Let $(A,B)$ be a biabduction problem and $(X,Y)$ a solution.
The intuition is that a model $(s,h)$ of both $A$ \emph{and} $B$ induces a total order
over the terms of $A,B$, dictating the form of the solution ($X, Y$).

\begin{figure*}
\begin{center}
\begin{tikzpicture}[y=-1cm]
\node (A) at (-1,0.25) {$A*X$};
\node (B) at (-1,1.25) {$B*Y$};

\tikzgap{0}{0}{1.8}{.5}{$c_1$}{$a_1 - 1$}
\tikzarray{1.8}{0}{5}{.5}{$a_1$}{$b_1$}
\tikzgap{5}{0}{7}{.5}{$b_1+1$}{$d_2$}

\tikzgap{8}{0}{10}{.5}{$c_3$}{$a_2 - 1$}
\tikzarray{10}{0}{12}{.5}{$a_2$}{$b_2$}
\tikzgap{12}{0}{14}{.5}{$b_2+1$}{$d_3$}

\tikzarray{0}{1}{2.1}{1.5}{$c_1$}{$d_1$}
\tikzgap{2.1}{1}{4.5}{1.5}{$d_1+1$}{$c_2 - 1$}
\tikzarray{4.5}{1}{7}{1.5}{$c_2$}{$d_2$}

\tikzarray{8}{1}{14}{1.5}{$c_3$}{$d_3$}

\end{tikzpicture}
\caption{Example showing solutions in Defn.~\ref{alg:biabduction}.
Arrays of $A,B$ are displayed as boxes and arrays in $X,Y$ as hatched rectangles.}
\label{fig:biabduction-example}
\end{center}
\end{figure*}
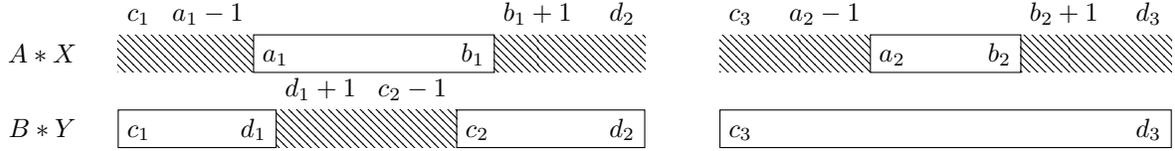

Consider \MakeUppercase fig.\nobreakspace \ref {fig:biabduction-example}, which depicts a biabduction instance $(A,B)$
and a solution $(X,Y)$ (in hatched pattern), where all array endpoints in $A,B$ are
totally ordered (on the horizontal axis). Using this order, we can compute
$X,Y$, by covering parts that $B$ requires but $A$ does not provide ($X$) and
by covering parts that $A$ requires but $B$ does not provide ($Y$).

We capture this intuition by (a) defining a $\pres$ formula $\beta(A,B$)
which is shown to be satisfiable whenever there is a solution for the biabduction
problem $(A,B)$ (\MakeUppercase defn.\nobreakspace \ref {defn:beta}, \MakeUppercase prop.\nobreakspace \ref {prop:solution-to-beta-sat}); (b) showing that if
$\beta(A,B)$ is satisfiable then there exists a formula $\Delta$
capturing the total order over the terms of $A,B$, which we call a \emph{solution seed} (\MakeUppercase defn.\nobreakspace \ref {defn:seed}, \MakeUppercase thm.\nobreakspace \ref {thm:beta-sat-implies-seed});
and (c) showing that if there is a solution seed $\Delta$ then we can
generate a solution $X,Y$ for the biabduction problem $(A,B)$
(\MakeUppercase defn.\nobreakspace \ref {alg:biabduction}, \MakeUppercase thm.\nobreakspace \ref {thm:seed-implies-solution}). These results and the way they compose are shown
in Figure~\ref{fig:biabd-overview}. 

\begin{figure}
\begin{center}
\small
\begin{tikzpicture}[node distance=3cm]
\node (A) [text width=4cm,align=center] {existence of biabduction solution for $(A,B)$};
\node (B) [text width=3cm,align=center,below left of=A] {satisfiability of $\beta(A,B)$};
\node (C) [text width=4cm,align=center,below right of=A] {existence of solution seed for $(A,B)$};

\draw[thick,->,bend right] (A) edge node [left] {\MakeUppercase prop.\nobreakspace \ref {prop:solution-to-beta-sat}} (B) ;
\draw[thick,->,bend right] (B) edge node [below] {\MakeUppercase thm.\nobreakspace \ref {thm:beta-sat-implies-seed}} (C);
\draw[thick,->,bend right] (C) edge node [right] {\MakeUppercase thm.\nobreakspace \ref {thm:seed-implies-solution}} (A);
\end{tikzpicture}
\end{center}
\caption{Results on quantifier-free biabduction.\label{fig:biabd-overview}}
\end{figure}
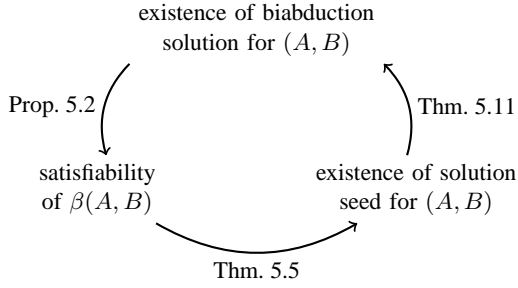
Finally, we show that the problem of finding a solution to a biabduction problem
is in $\NP$ and that our algorithm is complexity-optimal (\MakeUppercase prop.\nobreakspace \ref {prop:runtime_NP}).

\begin{definition}[The formula $\beta$]
\label{defn:beta}
Let $(A,B)$ be an instance of the biabduction problem, where
\[\begin{array}{l}
\textstyle A = \Pi : \bigsepstar^n_{i=1}\absarray{a_i}{b_i} * \bigsepstar^k_{i=1}\psto{t_i}{u_i} \\[15pt]
\textstyle B = \Pi' : \bigsepstar^m_{i=1}\absarray{c_i}{d_i} * \bigsepstar^{\ell}_{i=1}\psto{v_i}{w_i}
\end{array}\]
We define a formula $\beta(A,B)$ of $\pres$ as follows:
\[
\beta(A,B) \defeq
\begin{array}{l}
\gamma(A) \land \gamma(B)\ \land\\
\bigwedge_{j=1}^\ell\bigwedge_{i=1}^n (v_j<a_i \lor v_j>b_i)\ \land\\
\bigwedge_{i=1}^\ell\bigwedge_{j=1}^k (t_i\neq v_j\lor u_i=w_j )
\end{array}
\]
\end{definition}

\begin{proposition}
\label{prop:solution-to-beta-sat}
If the biabduction problem $(A,B)$ has a solution, then $\beta(A,B)$ is satisfiable.
\end{proposition}
\begin{proof} (Sketch) If $X,Y$ is a solution for the problem $(A,B)$, then
any model $s,h$ of $A*X$ (which exists, by assumption) is also a model of $B*Y$.
We then show that $s\models\beta(A,B)$, using Lemma\nobreakspace \ref {lem:gamma} for the first
conjunct of $\beta$, and the assumption that $A*X\models B*Y$ for the second
and third conjuncts.
\end{proof}

Given a biabduction problem of the form in Defn.~\ref{defn:beta},
we define a set of terms, $\abdterms{A}{B}$, by:
\[
\abdterms{A}{B} \defeq
\begin{array}{@{}l}
\trms{A} \cup \trms{B}\ \cup \\
\{ b_i + 1 \mid i\in[1,n]\} \cup \{d_i + 1\mid i\in[1,m]\}\ \cup\\
\{ t_i + 1 \mid i\in[1,k]\} \cup \{v_i + 1\mid i\in[1,\ell]\}
\end{array}
\]
where $\trms{-}$ denotes the set of all terms in a symbolic heap.

\begin{definition}[Solution seed]
\label{defn:seed}
A \emph{solution seed} for a biabduction problem $(A,B)$ in the form of Defn.~\ref{defn:beta}
is a pure formula \mbox{$\Delta = \bigwedge_{i\in I} \delta_i$}  such that:
\begin{enumerate}
\item\label{item:delta} $\Delta$ is satisfiable, and $\Delta \modelsPA \beta(A,B)$;
\item\label{item:terms} for all $i\in I$, the conjunct $\delta_i$ is of the form $(t < u)$ or $(t = u)$, where $t,u \in \abdterms{A}{B}$;
\item\label{item:order} for all $t,u \in \abdterms{A}{B}$, there exists $i\in I$ such that
$\delta_i$ is $(t < u)$ or $(u < t)$ or $(t = u)$.
\end{enumerate}
\end{definition}

\begin{lemma}
\label{lem:delta-ordering}
Let $\Delta$ be a solution seed for a biabduction problem $(A,B)$. $\Delta$ induces
a total order on $\abdterms{A}{B}$: for any $e,f\in\abdterms{A}{B}$,
$\Delta\modelsPA e < f$ or $\Delta\modelsPA e = f$ or  $\Delta\modelsPA f < e$.
%
\end{lemma}
This lemma justifies abbreviating
$\Delta \modelsPA e < f$ by $e \Dlt f$; $\Delta \modelsPA e \le f$ by $e\Dleq f$; and,
$\Delta \modelsPA e = f$ by $e \Deq f$.

\begin{theorem}
\label{thm:beta-sat-implies-seed}
If $\beta(A,B)$ is satisfiable,
then there exists a solution seed $\Delta$ for the biabduction problem $(A,B)$.
\end{theorem}
\begin{proof} (Sketch)
Supposing $s \modelsPA \beta(A,B)$, we define $\Delta$ as:
\[
\Delta \defeq
\textstyle\bigwedge_{\substack{e,f \in \abdterms{A}{B} \\ s(e) < s(f)}}\, e < f
\;\wedge\;
\bigwedge_{\substack{e,f \in \abdterms{A}{B} \\ s(e) = s(f)}}\, e = f .
\]%
We then show that $\Delta$ satisfies \MakeUppercase defn.\nobreakspace \ref {defn:seed}.
\end{proof}

We now present a way to compute a solution $(X,Y)$ given a solution seed $\Delta$.
They key ingredient is the $\arrcovalg$ algorithm, given in Fig.~\ref{alg:cover}. Intuitively,
$\arrcovalg$ takes a solution seed $\Delta$ and the endpoints of an $\absarray{c_j}{d_j}$ in $B$,
and constructs arrays for $X$ in such a way so that every model of
$A*X$ includes a submodel that satisfies $\absarray{c_j}{d_j}$.
To do this, arrays in $A$ contribute to the coverage of
$\absarray{c_j}{d_j}$ and, in addition, the newly created arrays do not overlap
with those of $A$ (or themselves) for reasons of consistency.

Note that in $\arrcovalg$ we sometimes need to generate terms denoting the
predecessor of the start of an array, even though there is no predecessor function
in $\pres$. We achieve this by introducing primed terms $a'_i$, and add pure constraints
that induce this meaning ($a_i+1 = a'_1$). This is done on demand by $\arrcovalg$
in order to avoid the risk of trying to decrement a zero-valued term, thus obtaining
an inconsistent formula.

\begin{figure*}
\removelatexerror
\small
\centering
\begin{tabular}{@{}c@{}c@{}}
\begin{minipage}{.475\textwidth}
\vspace{1cm}
\begin{function}[H]
\SetKwProg{Fn}{Function}{}{}
\SetCommentSty{textit}
\Fn{\arrcov{A,\Delta}{e}{f}}{
\KwData{a quantifier-free symbolic heap $A$; \phantom{xxxxxxxxxx} solution seed $\Delta$; terms $e$, $f$ in $\abdterms{A}{B}$}
\KwResult{quantifier-free symbolic heap}
\BlankLine
\tcp{work with $\mapsto$-abstraction of $A$}
\textbf{let} $\left(\Pi:\bigsepstar_{i=1}^{n+k}\absarray{{\hat a}_i}{{\hat b}_i}\right) = \abstr{A}$\label{ln:let-abstr}\;
\BlankLine
\If{$f \Dlt e$}
{
  \label{ln:u-less-than-t}
  \tcp{nothing to cover}
  \Return{\label{ln:return-emp}$\emp$}\;
}
\BlankLine
\If{$\exists i\in[1,n+k].\ {\hat a}_i \Dleq e \Dleq {\hat b}_i$}{
  \label{ln:t-covered}
  \tcp{left endpoint $e$ covered by $\absarray{{\hat a}_i}{{\hat b}_i}$}
  \Return{\label{ln:recurse-through}$\arrcov{A,\Delta}{{\hat b}_i+1}{f}$}\;
}
\BlankLine
\tcp{left endpoint $f$ not covered}
$E := \{ {\hat a}_j \mid e \Dlt {\hat a}_j \Dleq f \text{ for } j\in[1,n+k] \}$
\label{ln:def-e}\;
\If{$E = \emptyset$}
{
  \label{ln:e-empty}
  \tcp{no part of $\absarray{e}{f}$ covered}
  \Return{\label{ln:array-t-u}$\absarray{e}{f}$}\;
}
\BlankLine
\tcp{middle of $\absarray{e}{f}$ covered by $\absarray{{\hat a}_i}{{\hat b}_i}$}
${\hat a}_i := \min_\Delta(E)$\label{ln:i-minimal}\;
\Return{\label{ln:array-t-ai}$({\hat a'}_i + 1 = {\hat a}_i :\absarray{e}{{\hat a'}_i}) * \arrcov{A,\Delta}{{\hat b}_i+1}{f}$}\;
}
\end{function}
\end{minipage}
&
\begin{minipage}[c]{.525\textwidth}
\begin{function}[H]
\SetKwProg{Fn}{Function}{}{}
\SetCommentSty{textit}
\Fn{\ptocov{A,\Delta}{e}{f}}{
\textbf{let} $\left(\Pi : \bigsepstar^n_{i=1}\absarray{a_i}{b_i} * \bigsepstar^k_{i=1}\psto{t_i}{u_i}\right) = A$\;
\If{$\exists i\in[1,k].\ t_i \Deq e$}{
  \Return{\label{ln:pto-pto}$\emp$}\;
}
\If{$\exists i\in[1,n].\ a_i \Dleq e \Dleq b_i$}{
  \Return{\label{ln:pto-arr}$\emp$}\;
}
\label{ln:pto-notcovered}\Return{$\psto{e}{f}$}\;
}
\end{function}
\hrule
\begin{tikzpicture}[y=-1cm]
\node[anchor=north west,text width=8cm] (text) at (-2.75,-0.5)
{
\begin{itemize}[nosep]
\item \textbf{Arrays of $A$ / $B$} appear as boxes with indicated bounds.
\item \textbf{Arrays of $X$} appear in a hatched pattern.
\item \textbf{Recursive calls} appear as dashed boxes with parameters.
\item \textbf{Terms $a'_i$} are shown as $a_i - 1$ for readability.
\end{itemize}
};

\node (A1) at (-1,2.25) {$A*X$};
\node (B1) at (-1,3) {$B$};
\node (L1) at (-2.25,2.6) {Line~\ref{ln:recurse-through}:};

\tikzarray{0}{2}{2}{2.5}{$a_i$}{$b_i$}
\tikzdashedarray{2}{2.6}{5}{3.4}{}{$\arrcov{\sigma}{b_i+1}{u}$}
\tikzarray{1}{2.75}{5}{3.25}{$t$}{$u$}

\node (A2) at (-1,4.25) {$A*X$};
\node (B2) at (-1,5) {$B$};
\node (L2) at (-2.25,4.6) {Line~\ref{ln:array-t-u}:};

\tikzgap{1}{4}{4}{4.5}{$t$}{$u$}
\tikzarray{1}{4.75}{4}{5.25}{$t$}{$u$}

\node (A3) at (-1,6.25) {$A*X$};
\node (B3) at (-1,7) {$B$};
\node (L2) at (-2.25,6.6) {Line~\ref{ln:array-t-ai}:};

\tikzgap{0}{6}{1.5}{6.5}{$t$}{$a_i - 1$}
\tikzarray{1.5}{6}{3}{6.5}{$a_i$}{$b_i$}
\tikzdashedarray{3}{6.6}{6}{7.4}{}{$\arrcov{\sigma}{b_i+1}{u}$}

\tikzarray{0}{6.75}{6}{7.25}{$t$}{$u$}

\end{tikzpicture}
\end{minipage}
\end{tabular}
\caption{
Left: the function $\arrcov{A,\Delta}{e}{f}$.
Top right: the function $\ptocov{A,\Delta}{e}{f}$.
Bottom right: arrays of $A$, $B$, $X$ relevant to each \textbf{return} statement
in the $\arrcovalg$ function.}
\label{alg:cover}
\end{figure*}
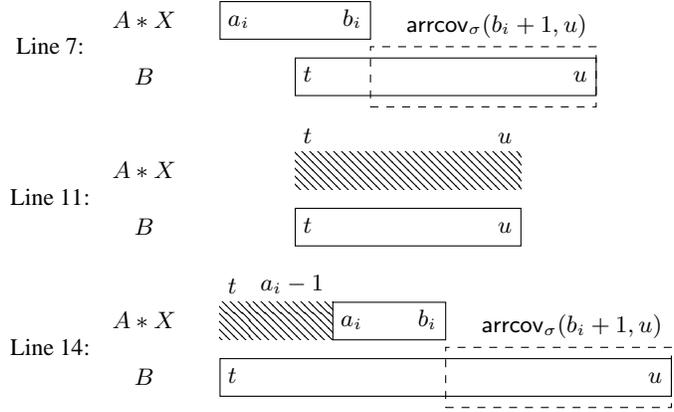

\begin{definition}[The formulas $X,Y$]
\label{alg:biabduction}
Let $\Delta$ be a solution seed for a biabduction problem $(A,B)$ in the form given in Defn.~\ref{defn:seed}.
The formulas $X,Y$ are defined as follows:
\[\begin{array}{@{}c@{}}
\displaystyle\Theta_X : F_X \defeq \hfill \\
\qquad\bigsepstar_{j=1}^{m} \arrcov{A,\Delta}{c_j}{d_j} * \bigsepstar_{j=1}^{\ell} \ptocov{A,\Delta}{v_j}{w_j} \\[15pt]
\displaystyle\Theta_Y : F_Y \defeq \hfill \\
\qquad \bigsepstar_{i=1}^{n} \arrcov{B,\Delta}{a_i}{b_i} * \bigsepstar_{i=1}^{k} \ptocov{B,\Delta}{t_i}{u_i} \\[15pt]
\hat \Delta \defeq \Delta\land\Theta_X\land\Theta_Y \\[5pt]
X  \defeq \hat \Delta : F_X \qquad
Y \defeq \hat \Delta : F_Y
\end{array}\]
\end{definition}

Every quantifier-free formula $A$ of $\ASL$ is \emph{precise} \cite{OHearn-Yang-Reynolds:04}
(by structural induction):
for any model $s,h$ there exists \emph{at most one} subheap $h'$ of $h$ such that $s,h'\models A$.
This motivates the following notation. We will write
$\cutheap{A}$ to denote the unique subheap $h'\subseteq h$ such that $s,h'\models A$, when it exists.

\begin{proposition}
\label{prop:termination}
Let $(A,B)$ be a biabduction problem of the form shown in Defn.~\ref{defn:seed}.
Let $\Delta$ be a solution seed and terms $e,f \in \abdterms{A}{B}$.
The call $\arrcov{A,\Delta}{e}{f}$:
\begin{enumerate}
\item always terminates, issuing up to $n+k$ recursive calls;
\item returns a formula (for some $q\in\nat$ and sets $I,J\subseteq\nat$)
\[\bigwedge_{i\in I} a_i = a_i'+1 \land
\bigwedge_{i\in J} t_i = t_i'+1  : \bigsepstar_{i=1}^q \absarray{l_i}{r_i}\]
where for all $i\in[1,q]$, $l_i \in \abdterms{A}{B}$;
\item for every $i\in[1,q]$, $\hat \Delta \modelsPA e \le l_i \le r_i \le f$;
\item for every $i\in[1,q-1]$, $\hat \Delta \models r_i<l_{i+1}$.
\end{enumerate}
\end{proposition}

\begin{lemma}
\label{lem:alg_sat}
Let $(A,B)$ be a biabduction instance, $\Delta$ a solution seed
and $X$ as in Defn.~\ref{alg:biabduction}.
Then, $A*X$ is satisfiable.
\end{lemma}
\begin{proof} (Sketch)
We first obtain a stack $s$ by unpacking \MakeUppercase defn.\nobreakspace \ref {defn:seed}. We extend it
to primed terms $a'_i$, $t'_i$, $c'_i$, $v'_i$ in a way that preserves satisfaction of $\Delta$.
Using this stack $s$, we then define appropriate heaps for the constituent parts of $A$ and $X$
and show that they are pairwise disjoint, thus constructing a heap that
satisfies $A*X$.
\end{proof}

\begin{definition}[The sequences $\Barrays,\Bptos,\Yarrays,\Yptos$]
\label{defn:subheaps}
Let $(A,B)$ be a biabduction problem, $\Delta$ a solution seed, $X,Y$ as
defined in \ref{alg:biabduction} and $s,h$ a model such that $s,h\models A*X$.
Then we define the following sequences $\Barrays,\Bptos,\Yarrays,\Yptos$ of subheaps of $h$,
such that:
\[\begin{array}{r@{\,\defeq\;}ll}
\Barrays_i& \cutheap{\absarray{c_i}{d_i}} & i\in[1,m] \\
\Bptos_i&\cutheap{\psto{v_i}{w_i}} & i\in[1,\ell]\\
\Yarrays_i &\cutheap{\arrcov{B,\Delta}{a_i}{b_i}} & i\in[1,n]\\
\Yptos_i &\cutheap{\ptocov{B,\Delta}{t_i}{u_i}} & i\in[1,k]
\end{array}\]
\end{definition}

\begin{lemma}
\label{lem:seqs-well-defined}
\label{lem:totality}
\label{lem:disjointness}
All heaps in $\Barrays,\Bptos,\Yarrays,\Yptos$ exist (that is, they are well-defined). Also,
\begin{enumerate}
\item For any sequence of heaps $\mathcal{S}$ of $\Barrays$,$\Bptos$,$\Yarrays$,$\Yptos$,
and any distinct $i,j\in[1,|\mathcal{S}|]$, $\mathcal{S}_i\hash\mathcal{S}_j$.
\item For any two distinct sequences of heaps $\mathcal{S},\mathcal{T}$ of $\Barrays$,
$\Bptos$, $\Yarrays$, $\Yptos$,
and any $i,j$, 
$\mathcal{S}_i\hash\mathcal{T}_j$.
\item $
\dom{h} \subseteq
\displaystyle\bigcup_{i=1}^m \Barrays_i \cup
\bigcup_{i=1}^\ell \Bptos_i \cup
\bigcup_{i=1}^n \Yarrays_i \cup
\bigcup_{i=1}^k \Yptos_i
$.
\end{enumerate}
\end{lemma}

%
%
%

\begin{theorem}
\label{thm:seed-implies-solution}
Given a solution seed $\Delta$ for the biabduction problem $(A,B)$,
the formulas $X$ and $Y$, as computed by Defn.~\ref{alg:biabduction},
form a solution for that instance.
\end{theorem}

\begin{proof}
That $(X,Y)$ is a solution means that $A*X$ is satisfiable and that $A*X\models B*Y$.
The first requirement is fulfilled by Lemma~\ref{lem:alg_sat}. Here, we show the second.

Let $s,h$ be a model of $A*X$.
We need to show that $s,h\models B*Y$. Using Defn.~\ref{alg:biabduction},
we have:
\[A*X = \Pi \land \hat\Delta : F_{A*X} \quad\text{ and }\quad
B*Y = \Pi' \land \hat\Delta : F_{B*Y}
\]
It is easy to see that $s\models \Pi'\land\hat\Delta$: by assumption, $s\models \hat\Delta$,
and as $\hat\Delta\models\Delta$ (\MakeUppercase defn.\nobreakspace \ref {alg:biabduction}) and $\Delta\modelsPA \gamma(B)$
(\MakeUppercase defn.\nobreakspace \ref {defn:seed}), it follows that $s\models \Pi'$ as well (\MakeUppercase defn.\nobreakspace \ref {defn:gamma}).

It remains to show that $s,h\models F_{B*Y}$.
Recall that \mbox{$F_{B*Y} = F_B * F_Y$} and that
\[\begin{array}{r@{\,}l}
F_B &= \displaystyle\bigsepstar^m_{i=1}\absarray{c_i}{d_i} * \bigsepstar^{\ell}_{i=1}\psto{v_i}{w_i}\\
F_Y & = \displaystyle\bigsepstar_{i=1}^{n} \arrcov{B,\Delta}{a_i}{b_i} * \bigsepstar_{i=1}^{k} \ptocov{B,\Delta}{t_i}{u_i}
\end{array}\]
We will do this by (a) defining a subheap $h' \subseteq h$ for each atomic formula $\sigma$ in $F_{B*Y}$,
such that $s,h'\models\sigma$. Having done this we will need
(b) to show that all such subheaps are disjoint, and that (c) their disjoint union equals $h$.

The sequences $\Barrays,\Bptos,\Yarrays,\Yptos$ from \MakeUppercase defn.\nobreakspace \ref {defn:subheaps}, by construction,
fulfil requirement (a) above, given they are well-defined as guaranteed by Lemma\nobreakspace \ref {lem:seqs-well-defined} (main statement).
Requirement (b) is covered by items 1 and 2 of Lemma\nobreakspace \ref {lem:disjointness}.
Finally, requirement (c) is covered by item 3 of Lemma\nobreakspace \ref {lem:totality}.
\end{proof}

\begin{remark}
The solutions obtained via \MakeUppercase defn.\nobreakspace \ref {alg:biabduction} are constructed from
terms in $\abdterms{A}{B}$. This is syntactically optimal in the sense that
$X,Y$ are as `symbolic' as $A,B$ are.

Our solutions are potentially stronger than required.
Applying \MakeUppercase defn.\nobreakspace \ref {alg:biabduction} to Example\nobreakspace \ref {ex:biabduction} gives us
several solutions, corresponding to the number of ways
$\relarray{\Cb}{0}{\Cn-1}$ and $\relarray{\Cm}{0}{\Ck-1}$ can be situated in
memory in relation to each other. However, it can be seen that some solutions
can be merged into one, weaker solution. For instance,
\[\begin{array}{r@{\,}l}
X_1 & = \Cm+\Ck\leq\Cb : \relarray{\Cm}{0}{\Ck-1} \\
X_2 & = \Cb+\Cn\leq\Cm : \relarray{\Cm}{0}{\Ck-1}
\end{array}\]
can be merged into the more natural $X =\relarray{\Cm}{0}{\Ck-1}$.

Our method is, also, complete in the following sense.
Suppose $(X,Y)$ is a solution that does not impose a total order over $\abdterms{A}{B}$.
Then, there exists a solution $(X',Y')$ computable by our method, such that
$X'\models X$ and $Y'\models Y$.
\end{remark}

\subsection{Complexity of quantifier-free biabduction in $\ASL$}
\label{sec:qfree-biabduction-bounds}

\begin{lemma}
\label{lem:delta-polytime}
Let $(A,B)$ be a biabduction instance and
$\Delta$ a formula satisfying \MakeUppercase conditions\nobreakspace \ref {item:terms} and\nobreakspace  \ref {item:order} of
\MakeUppercase defn.\nobreakspace \ref {defn:seed}. Let $\Gamma=\bigwedge\bigvee \pi$ be a formula where
$\pi$ is of the form $t< u$ or $t = u$ and $t,u\in\abdterms{A}{B}$. Then,
checking $\Delta\modelsPA \Gamma$ is in $\PTIME$.
\end{lemma}

\begin{proposition}
\label{prop:runtime_NP}
Deciding if there is a solution for a biabduction problem $(A,B)$,
and constructing it if it exists, can be done in $\NP$.
\end{proposition}
\begin{proof}(Sketch)
We guess a total order over $\abdterms{A}{B}$ and a polynomially-sized assignment of values $s$
(\cite[Theorem~6]{Scarpellini:84})
to all terms in $\abdterms{A}{B}$. We convert this order to a formula $\Delta$
and check if $s\models\Delta$ (thus showing the satisfiability of $\Delta$) and
whether $\Delta\models\beta(A,B)$. If all these conditions hold, we use
\MakeUppercase defn.\nobreakspace \ref {alg:biabduction} and obtain formulas $X,Y$. By \MakeUppercase prop.\nobreakspace \ref {prop:termination}
and Lemma\nobreakspace \ref {lem:delta-polytime} this process runs in $\PTIME$.
\end{proof}

We establish $\NP$-hardness of quantifier-free biabduction by reduction from the
3-partition problem, similarly to satisfiability in Section~\ref{sec:satisfiability}.

\begin{definition}\label{d-3part-to-biabd}
 Similar to Definition~\ref{defn:3part_to_sat}, given an instance
 \mbox{$(B,\cal{S})$} of the \mbox{$3$-}partition problem,
 we define corresponding symbolic heaps $\widetilde{A}_{B,{\cal S}}$
 and $\widetilde{B}_{B,{\cal S}}$,
 such that both are satisfiable, quantifier-free,
 \mbox{$\mapsto$-}free  and in two-variable
 form. First, $\widetilde{A}_{B,\cal{S}}$ is:
$$
\bigwedge_{i=1}^{m} (d_{i+1}= d_{i}+B+1) \colon
 \bigsepstar_{i=1}^{m+1} \relarray{d_i}{0}{0}\ .
$$%
 Clearly, $\widetilde{A}_{B,{\cal S}}$ is satisfiable, and
 the variables $d_i$ act as single-cell delimiters
 between memory chunks of length~$B$:
\newcount \DBL
 \DBL=14\multiply\DBL by 18%
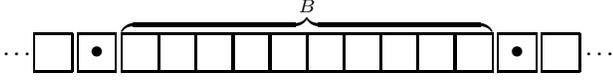
$$\begin{picture}(\DBL,25)%
  \ \raisebox{6pt}{\ldots}%
  \framebox(14,14)[c]{}
  \framebox(14,14)[c]{$\bullet$}
$\overbrace{%
  \framebox(14,14)[c]{}%
  \framebox(14,14)[c]{}%
  \framebox(14,14)[c]{}%
  \framebox(14,14)[c]{}%
  \framebox(14,14)[c]{}%
  \framebox(14,14)[c]{}%
  \framebox(14,14)[c]{}%
  \framebox(14,14)[c]{}%
  \framebox(14,14)[c]{}%
  \framebox(14,14)[c]{}%
}^{B}$
  \framebox(14,14)[c]{$\bullet$}
  \framebox(14,14)[c]{}%
  \ \raisebox{6pt}{\ldots}%
\end{picture}$$%
\vspace*{1ex}

\noindent
We define $\widetilde{B}_{B,\cal{S}}$ as the following symbolic heap
 (essentially a relaxed version of\/ $A_{B,\cal{S}}$
 from Definition~\ref{defn:3part_to_sat}):
\[\begin{array}{c}%
  \displaystyle\bigwedge_{i=1}^{m} (d_{i+1}> d_{i}) \wedge
  \bigwedge_{j=1}^{3m} (d_{1}\leq a_j)\wedge(a_j+k_j< d_{m+1}) \colon
 \\
 \displaystyle\bigsepstar_{i=1}^{m+1} \relarray{d_i}{0}{0}
 * \displaystyle\bigsepstar_{j=1}^{3m} \relarray{a_j}{1}{k_j}\ .
\end{array}\]%
 $\widetilde{B}_{B,{\cal S}}$ is satisfiable,
 since the ``liberal'' \mbox{$\bigwedge_{i=1}^{m} (d_{i+1}> d_{i})$}
 allows us to allocate arrays of any length $k_j$ anywhere
 in the unbounded space between the delimiters $d_1$ and $d_{m+1}$.
 E.g.,

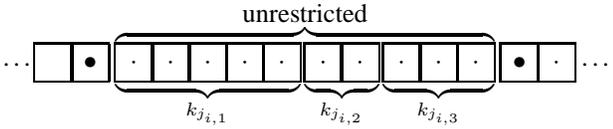
\newcount \DBL
 \DBL=14\multiply\DBL by 18%
$$\begin{picture}(\DBL,25)%
  \ \raisebox{6pt}{\ldots}%
  \framebox(14,14)[c]{}%
  \framebox(14,14)[c]{$\bullet$}
$\overbrace{%
\underbrace{%
  \framebox(14,14)[c]{$\cdot$}%
  \framebox(14,14)[c]{$\cdot$}%
  \framebox(14,14)[c]{$\cdot$}%
  \framebox(14,14)[c]{$\cdot$}%
  \framebox(14,14)[c]{$\cdot$}%
}_{k_{j_{i,1}}} 
\underbrace{%
  \framebox(14,14)[c]{$\cdot$}%
  \framebox(14,14)[c]{$\cdot$}%
}_{k_{j_{i,2}}}
\underbrace{%
  \framebox(14,14)[c]{$\cdot$}%
  \framebox(14,14)[c]{$\cdot$}%
  \framebox(14,14)[c]{$\cdot$}%
}_{k_{j_{i,3}}}
}^{\mbox{unrestricted}}$
  \framebox(14,14)[c]{$\bullet$}%
  \framebox(14,14)[c]{$\cdot$}%
  \ \raisebox{6pt}{\ldots}%
\end{picture}$$%
\vspace*{1ex}
\end{definition}

\noindent The correctness of our reduction
 is established by the next lemma.

\begin{lemma}\label{l-beta-gamma}
Let\/ $A_{B,\cal{S}}$ be the symbolic heap given
 by Definition~\ref{defn:3part_to_sat}.
 Then we have the Presburger equivalence
\[
\beta(\widetilde{A}_{B,{\cal S}},\widetilde{B}_{B,{\cal S}})
\equiv \gamma(A_{B,{\cal S}})\ .
\]%
\end{lemma}

\begin{proof}(Sketch)
Follows from Defns.~\ref{d-3part-to-biabd},~\ref{defn:beta} and~\ref{defn:gamma}.
\end{proof}

\begin{theorem}\label{t-biadjunction-NP-hard}
 The biabduction problem for $\ASL$ is $\NP$-hard, even for problem
 instances \mbox{$(A,B)$} such that $A,B$ are satisfiable,
 quantifier-free, \mbox{$\mapsto$-}free and in two-variable form.
\end{theorem}

\begin{proof}
By reduction from the \mbox{$3$}-partition problem (see
Section~\ref{sec:satisfiability}). Given an instance $(B,\cal{S})$ of
this problem, let $A_{B,{\cal S}}$, $\widetilde{A}_{B,{\cal S}}$, and $\widetilde{B}_{B,{\cal S}}$ be the symbolic heaps given by
Defns.~\ref{defn:3part_to_sat} and~\ref{d-3part-to-biabd}.  Note that $\widetilde{A}_{B,{\cal S}},
\widetilde{B}_{B,{\cal S}}$ are satisfiable, quantifier- and $\mapsto$-free,
 and in two variable form.  Then we have
\[\begin{array}{@{}c@{\;}l@{}l@{}}
& \mbox{$\exists$ complete 3-partition on~$\cal{S}$\/ (w.r.t.~$B$)} \\
\iff & \mbox{$A_{B,{\cal S}}$ is satisfiable} & \mbox{(Lemma~\ref{lem:3part_to_sat})} \\
\iff & \mbox{$\gamma(A_{B,{\cal S}})$ is satisfiable} & \mbox{(Lemma~\ref{lem:gamma})} \\
\iff &  \mbox{$\beta(\widetilde{A}_{B,{\cal S}},\widetilde{B}_{B,{\cal S}})$ is satisfiable}
& \mbox{(Lemma~\ref{l-beta-gamma})} \\
\iff & \mbox{$\exists$ biabduction solution for $(\widetilde{A}_{B,{\cal S}}, \widetilde{B}_{B,{\cal S}})$}
\\ & \mbox{(Prop.~\ref{prop:solution-to-beta-sat} / Thm.~\ref{thm:beta-sat-implies-seed} / Thm.~\ref{thm:seed-implies-solution})}
\end{array}\]
This completes the reduction.
\end{proof}

\subsection{Biabduction for $\ASL$ with quantifiers}
\label{sec:quantified-biabduction}

Here we show two complementary results about biabduction in the case where $B$ contains existential quantifiers.  First, we show that if the quantifiers are appropriately restricted, then the biabduction problem is equivalent to the quantifier-free case (and thus $\NP$-solvable). At the same time, if \emph{no} restrictions are placed on the quantifiers, then the problem becomes $\Pi^P_2$-hard in Stockmeyer's \emph{polynomial-time hierarchy}~\cite{Stockmeyer:77}.

\begin{proposition}
\label{prop:biabd_qf}
Let $A$ be quantifier-free, and let $B$ be such that no variable appearing in the RHS of a $\mapsto$ formula is existentially bound.
Then an instance $(A,B)$ of the biabduction problem for $\ASL$ has a solution
if and only if $(A,\qf{B})$ has a solution.
\end{proposition}

\begin{proof}(Sketch)
Let $B = \exists\vec{z}.\ Q$, where $Q=\qf{B}$ is quantifier-free. The $(\Leftarrow)$ direction is trivial.  For the ($\Rightarrow$) direction, suppose $A * X$ is satisfiable and $A * X \models (\exists\vec{z}.\ Q) * Y$. Since the free variables in $Y$ are disjoint from $\vec{z}$, this can be rewritten as $A * X \models \exists\vec{z}.\ (Q * Y)$.  There is a stack-heap pair $(s,h)$ such that $s,h \models A*X$ and, moreover, we may choose $h$ such that $h(x) \neq s(w)$ for all formulas of the form $\psto{v}{w}$ occurring in $Q * Y$, and for all $x$ such that $\absarray{a}{b}$ occurs in $A * X$ and $s(a) \leq s(x) \leq s(b)$.

Now, since $A * X \models (\exists\vec{z}.\ Q) * Y$, we get $s,h \models \exists\vec{z}.\ (Q * Y)$, meaning that $s[\vec{z} \mapsto \vec{m}], h \models Q * Y$ for some $\vec{m}$.
We define an extension of the symbolic heap $X$ by:
\[
X' \defeq \big(\textstyle\bigwedge_{x \in \fv{A,X,Q,Y}} x =s[\vec{z}\mapsto \vec{m}](x) \big) * X
\]
We then verify that $(X',Y)$ is a solution for $(A,Q)$. Our choice of $h$ is crucial in verifying that $A * X \models Q * Y$.
\end{proof}

The construction of $h$ in the proof of Prop.~\ref{prop:biabd_qf} affords some insight into the reasons for the restrictions on our existential quantifiers: the contents of the arrays in $h$ are chosen different to the data values occurring in the $\mapsto$-formulas in $B$.  If any such values are quantified, this may become impossible.  Indeed, $X = Y = \emp$ is a trivial biabduction solution for $\absarray{x}{x} * X \models (\exists y.\ x \mapsto y) * Y$, but no solution exists if we remove the quantifier.

In order to obtain the $\Pi^P_2$ lower bound for biabduction with unrestricted quantifiers,
 we exhibit a reduction from the following \emph{colourability} problem, from~\cite{Ajtai-etal:00}.

\begin{problem}
{$2$-round $3$-colourability problem}
Given an undirected graph \mbox{$G=(V,E)$} with $n$~vertices
 $v_1$, \dots, $v_k$, $v_{k+1}$, \dots $v_n$,
 let $v_1$, $v_2$, \dots, $v_k$ be its {\em leaves}.
 The 2-round 3-colourability problem is to decide
 whether every \mbox{$3$-}colouring of the leaves can be extended to a
\mbox{$3$-}colouring of the whole graph $G$,
 such that no two adjacent vertices share the same colour.
\end{problem}

\begin{definition}
\label{d-colour-biabduct}
Let \mbox{$G=(V,E)$} be an undirected graph with $n$~vertices and $k$~leaves. We define a pair of symbolic heaps, intended to simulate the colourability game on~$G$: \ \ $A_{G}$ will encode an arbitrary \mbox{$3$-}colouring of the leaves, and\/ $B_{G}$ will encode a perfect \mbox{$3$-}colouring of the whole~$G$.

We use $c_{i,1}$ to denote one of the colours, $1$, $2$, or~$3$, the vertex~$v_i$ is marked by. We mark also each edge \mbox{$(v_i,v_j)$} by $\widetilde{c_{ij}}$, the colour ``complementary'' to $c_{i,1}$ and\/ $c_{j,1}$.

As for the leaves\/~$v_i$, we introduce \mbox{$k$} distinct numbers
$d_1, \ldots, d_k$ so that the value $c_i$ stored in the
location~$d_i$ can be used subsequently to identify the colour
$c_{i,1}$ marking $v_i$, e.g., with the help of \mbox{$(c_{i,1}-1 \equiv c_i\ (\bmod {3}))$}~.

To encode the fact that no two adjacent vertices $v_i$ and $v_j$ share the same colour, we use $c_{i,1}$, $c_{j,1}$, and $\widetilde{c_{ij}}$ as the addresses for three consecutive cells within a memory chunk of length~$3$ given by $\relarray{e_{ij}}{1}{3}$, which forces these colours to form a {\em permutation} of \mbox{$(1,2,3)$}.
 (The base-offset addresses $e_{ij}$ are chosen to ensure
 that all the arrays in question are disjoint.)

\noindent Concretely, we define $A_G$ to be the following symbolic heap:
\[ \textstyle\bigsepstar_{i=1}^{k} \relarray{d_i}{1}{1} *
\textstyle\bigsepstar_{(v_i,v_j)\in E} 
\,\relarray{e_{ij}}{1}{3}\ .\]%
\noindent We define $B_G$ as follows:
\[\begin{array}{@{}l} \exists \vec{z}.\ \big(
\textstyle\bigwedge_{i=1}^{n}(1\leq{c_{i,1}}\leq 3) \wedge
\textstyle\bigwedge_{(v_i,v_j)\in E}(1\leq \widetilde{c_{ij}}\leq 3)
\\[1ex]
 \wedge\ \
\textstyle\bigwedge_{i=1}^{k}\,(c_{i,1}-1 \equiv c_i\ (\bmod {3}))
\colon \\[1ex] 
 \textstyle\bigsepstar_{i=1}^{k} d_i \mapsto c_i \ *
 \textstyle\bigsepstar_{(v_i,v_j)\in E} 
 \relarray{e_{ij}}{c_{i,1}}{c_{i,1}}
\\[2ex]
 * \textstyle\bigsepstar_{(v_i,v_j)\in E}\,
   \relarray{e_{ij}}{c_{j,1}}{c_{j,1}}
   *\relarray{e_{ij}}{\widetilde{c_{ij}}}{\widetilde{c_{ij}}} \big).
\end{array}\]%
 where the existentially quantified variables $\vec{z}$
 are all variables occurring in\/~$B_{G}$ that are not mentioned
 explicitly in\/~$A_{G}$.
\end{definition}

\begin{lemma}\label{l-colour-biabduct}
Let $G$ be a
\mbox{$2$-}round \mbox{$3$-}colouring instance. The biabduction problem
\mbox{$(A_G,B_G)$} has a solution iff there is a winning
strategy for 
 colouring~$G$, where $A_G$ and $B_G$ are
 the symbolic heaps given by Defn.~\ref{d-colour-biabduct}.  
\end{lemma}

\begin{theorem}\label{t-biabduct-P2-hard} The biabduction problem
\mbox{$(A,B)$} for $\ASL$ is \mbox{$\Pi_2^P$-}hard, even if $A$ is
quantifier-free and \mbox{$\mapsto$-}free.
\end{theorem}

\begin{proof}
Follows from Lemma~\ref{l-colour-biabduct}.
\end{proof}

\section{Entailment}
\label{sec:entailment}

In this section, we investigate the \emph{entailment} problem for $\ASL$.  We establish an upper bound of $\Pi^{\mathrm{EXP}}_1$ in the \emph{weak $\EXP$ hierarchy}~\cite{Hartmanis85} via an encoding into $\Pi^0_2$ $\pres$, and a lower bound of $\Pi^P_2$ in the \emph{polynomial-time hierarchy}~\cite{Stockmeyer:77}. Moreover, for quantifier-free entailments, we show that the problem becomes $\CoNP$-complete.

\begin{problem}{Entailment problem for $\ASL$}
Given symbolic heaps $A$ and $B$, decide whether $A \models B$.

As in the biabduction problem, $A$ may be considered quantifier-free, but the existential quantifiers in $B$ may not mention any variable appearing in the RHS of a $\mapsto$-formula.
\end{problem}

The intuition underlying our encoding of entailment into Presburger arithmetic is as follows: There exists a countermodel for $A \models B$ iff there exists a stack $s$ that induces a model for $A$ (captured by $\gamma(A)$ from Defn.~\ref{defn:gamma} / Lemma~\ref{lem:gamma}) and, for every instantiation of the existentially quantified variables in $B$ (say $\vec{z}$), one of the following holds under $s$:
\begin{enumerate}
\item the quantifier-free body $\qf{B}$ of $B$ becomes unsatisfiable (captured by $\neg\gamma(\qf{B})$); or
\item some heap location is covered by an array or pointer in $A$, but not by any array or pointer in $B$, or vice versa; or
\item the LHS of some pointer in B is covered by an array in $A$ (and therefore we can choose the contents of the array different to the ``correct'' data contents of the pointer); or
\item some pointer in $B$ is covered by a pointer in $A$, but their data contents disagree.
\end{enumerate}

Similar to Prop.~\ref{prop:biabd_qf}, this intuition also explains the reason for our restriction on existential quantification in the entailment problem: if we are allowed to quantify over the RHS of $\mapsto$ formulas, then item 3 above might or might not be sufficient to construct a countermodel.   For example, there is a countermodel for $\absarray{x}{x} \models \exists y.\ y \leq 3 : x \mapsto y$, and for $\absarray{x}{x} \models x \mapsto y$,  but not for $\absarray{x}{x} \models \exists y.\ \psto{x}{y}$.

\begin{definition}
\label{defn:entail_phi}
Let $A$ and $B$ be $\mapsto$-free symbolic heaps, with spatial parts as follows:
\[\begin{array}{rl}
A: & \absarray{a_1}{b_1} * \ldots * \absarray{a_n}{b_n} \\
B: & \absarray{c_1}{d_1} * \ldots * \absarray{c_m}{d_m}
\end{array}\]
Then we define the formula $\phi(A,B)$ of $\pres$ to be
\[
\exists x.\ \textstyle\bigvee_{i=1}^{n}a_i \leq x \leq b_i \wedge \bigwedge_{j=1}^{m}(x < c_j) \vee (x > d_j)\ ,
\]
where $x$ is a fresh variable. We lift $\phi(-,-)$ to arbitrary symbolic heaps by $\phi(A,B) = \phi(\abstr{\qf{A}},\abstr{\qf{B}})$, i.e. by ignoring quantifiers and abstracting pointers to arrays using $\abstr{-}$ from Defn.~\ref{defn:abstraction}.
\end{definition}

\begin{lemma}\label{l-kill-phi}
 We can rewrite \mbox{$\phi(A,B)$}
 as a quantifier-free formula at only polynomial cost.
\end{lemma}

\begin{definition}
\label{defn:entail_coding}
Let $A$ and $B$ be symbolic heaps with $A$ quantifier-free:
\[\begin{array}{rl}
A: & \Pi : \bigsepstar^n_{i=1}\absarray{a_i}{b_i} * \bigsepstar^k_{i=1}\psto{t_i}{u_i} \\
B: & \exists \vec{z}.\ \Pi' : \bigsepstar^m_{j=1}\absarray{c_j}{d_j} * \bigsepstar^\ell_{j=1}\psto{v_j}{w_j}
\end{array}\]
where the existentially quantified variables $\vec{z}$ are disjoint from all variables in $A$.  We define formulas $\psi_1(A,B)$, $\psi_2(A,B)$ and $\chi(A,B)$ of $\pres$ as follows:
\[\begin{array}{@{}l@{\hspace{0.2cm}}l@{}}
\psi_1(A,B) =  & \bigvee_{i=1}^n\bigvee_{j=1}^{\ell} a_i \leq v_j \leq b_i\ , \\
\psi_2(A,B) = & \bigvee_{i=1}^k\bigvee_{j=1}^{\ell}(t_i=v_j) \wedge (u_i \neq w_j) \mbox{ , and} \\[0.5ex]
\chi(A,B) = & \gamma(A) \wedge \forall\vec{z}.\big(\neg\gamma(\qf{B}) \vee \phi(A,B)\ , \\
&\quad \mathrel{\vee} \phi(B,A) \vee \psi_1(A,B) \vee \psi_2(A,B)\big)
\end{array}\]%
where $\gamma(-)$ is the encoding of satisfiability (\MakeUppercase defn.\nobreakspace \ref {defn:gamma}), and $\phi(-,-)$ is given by Defn.~\ref{defn:entail_phi}.
\end{definition}

\begin{lemma}
\label{lem:entail_coding}
For any instance $(A,B)$ of the $\ASL$ entailment problem above, and for any stack $s$,
\[
s \modelsPA \chi(A,B) \;\iff\; \exists h.\ s,h \models A\ \mbox{ and } s,h\not\models B\ .
\]%
\end{lemma}

\begin{theorem}
\label{thm:ent_decide}
Entailment for $\ASL$ is in $\Pi^{\mathrm{EXP}}_1$.  If the number of variables in $A,B$ is fixed then the problem is in $\Pi^P_2$, and if $B$ is quantifier-free then the problem is in $\CoNP$.
\end{theorem}

\begin{proof}
Let $A$ and $B$ be symbolic heaps with $A$ quantifier-free.  Letting $\vec{x}$ be a list of all free variables in $A$ and $B$, we observe that $\exists \vec{x}.\ \chi(A,B)$ is a $\Sigma^0_3$ $\pres$ sentence of size polynomial in the size of $A$ and $B$.
By Lemma~\ref{lem:entail_coding}, we have that $\exists \vec{x}.\ \chi(A,B)$ is satisfiable if and only if $A \not\models B$. Therefore, $A \models B$ if and only if the $\Pi^0_3$ $\pres$ sentence $\forall \vec{x}.\ \neg\chi(A,B)$ is satisfiable.

However, according to Lemma~\ref{l-kill-phi}, we can eliminate the existential quantifier from the subformulas $\phi(A,B)$ and $\phi(B,A)$ inside $\chi(A,B)$, still at only polynomial cost.  Writing $\chi'(A,B)$ for the formula so obtained,
$\forall \vec{x}. \neg\chi'(A,B)$ then becomes a $\Pi^0_2$ sentence.

Satisfiability in $\Pi^0_2$ Presburger arithmetic is in $\Pi^{\mathrm{EXP}}_1$~\cite{Haase:14}.  If the set of variables in $A$ and $B$ has fixed size $k$, then the decision sentence above has exactly $k+1$ quantifiers, in which case satisfiability is in $\Pi^P_2$~\cite{Gradel:88}. Finally, if $B$ is quantifier-free, the decision sentence is a $\Pi^0_1$ formula and so can be decided in $\CoNP$ time~\cite{Stockmeyer:77}.
\end{proof}

In order to obtain the $\Pi^P_2$ lower bound for entailment, we exhibit a reduction from the same colourability problem as in Section~\ref{sec:quantified-biabduction}.

\begin{definition}
\label{defn:colour_to_entail} (cf.~Definition~\ref{d-colour-biabduct})
Let \mbox{$G=(V,E)$} be an undirected graph with $n$~vertices and $k$~leaves. To simulate the colourability game on~$G$, we define a pair of \mbox{$\mapsto$-}free symbolic heaps: a quantifier-free \mbox{$A_{G}$}, to encode an arbitrary \mbox{$3$-}colouring of the leaves, and an existentially quantified \mbox{$B_{G}$}, to encode a perfect \mbox{$3$-}colouring of the whole~$G$.

We use $c_{i,1}$ to denote the colour the vertex~$v_i$ is marked by. We mark also each edge \mbox{$(v_i,v_j)$} by
 $\widetilde{c_{ij}}$, ``complementary'' to $c_{i,1}$ and\/ $c_{j,1}$.

We encode the fact that no two adjacent vertices $v_i$ and $v_j$ share the same colour in accordance with Definition~\ref{d-colour-biabduct}. (The numbers $e_{ij}$ are chosen to ensure that all the arrays in question are disjoint.)

\noindent Concretely, we define $A_G$ to be the following symbolic heap:
\[\textstyle\bigwedge_{i=1}^{k}(1\leq{c_{i,1}}\leq 3)\colon
\textstyle\bigsepstar_{(v_i,v_j)\in E}\, \relarray{e_{ij}}{1}{3}\ .\]%

\noindent We define $B_G$ as follows:
\[\begin{array}{@{}l} \exists \vec{z}.\ \big(
\textstyle\bigwedge_{i=1}^{n}(1\leq{c_{i,1}}\leq 3) \wedge
\textstyle\bigwedge_{(v_i,v_j)\in E}(1\leq \widetilde{c_{ij}}\leq 3)
\colon \\[1ex]
 \textstyle\bigsepstar_{(v_i,v_j)\in E}\,
 \relarray{e_{ij}}{c_{i,1}}{c_{i,1}} \\
 * \textstyle\bigsepstar_{(v_i,v_j)\in E}\,
   \relarray{e_{ij}}{c_{j,1}}{c_{j,1}}
   *\relarray{e_{ij}}{\widetilde{c_{ij}}}{\widetilde{c_{ij}}} \big).
\end{array}\]%
 where the existentially quantified variables $\vec{z}$
 are all variables occurring in\/~$B_{G}$ that are not mentioned
 explicitly in\/~$A_{G}$.
\end{definition}

\begin{lemma}
\label{lem:colour_to_entail}
Let $G$ be a 2-round \mbox{3-}colouring instance,
 and let $A_G$ and $B_G$ be the symbolic heaps given
 by Defn.~\ref{defn:colour_to_entail}. Then, we have
\[
A_G \models B_G \;\iff\;
 \exists\ \mbox{winning strategy for colouring $G$.}
\]%
\end{lemma}

\begin{theorem}
\label{t-entail-P2-hard}
 The entailment problem \mbox{$A\models B$} is $\Pi_2^P$-hard,
 even when all variables are bounded by\/~$3$,
 $A$~is quantifier-free, and
 $A,B$ are \mbox{$\mapsto$-}free symbolic heaps
 in two-variable form.
 Moreover, the entailment problem is $\CoNP$-hard even for
 quantifier-free symbolic heaps in two-variable form.
\end{theorem}

\begin{proof}
 For the general case, Definition~\ref{defn:colour_to_entail} and
 Lemma~\ref{lem:colour_to_entail} establish a reduction from
 the 2-round 3-colourability problem,
 which is $\Pi^P_2$-hard~\cite{Ajtai-etal:00}.

For the quantifier-free case, the upper bound is immediate by Thm.~\ref{thm:ent_decide}.  For the lower bound, consider the entailment
\[
  A_{B,{\cal S}}\models x < x : \emp
\]
where $(B,\cal{S})$ is an instance of the 3-partition problem (see
Section~\ref{sec:satisfiability}) and $A_{B,{\cal S}}$ is the symbolic heap in two-variable form constructed
in Defn.~\ref{defn:3part_to_sat}.  Using Lemma~\ref{lem:3part_to_sat}, this entailment is valid iff there is \emph{no} complete \mbox{$3$-}partition on $\cal{S}$ w.r.t. $B$, which is a $\CoNP$-hard problem.
\end{proof}

In the general case, there is a complexity gap between our upper and lower bounds for entailment: $\Pi^\mathrm{EXP}_1 = \mathsf{coNEXP}^{\NP}$ versus $\Pi^P_2 = \CoNP^{\NP}$, respectively. It seems plausible that the lower bound is at least $\EXPTIME$: however, an encoding of, e.g., $\Pi^2_0$ Presburger arithmetic in $\ASL$ does not seem straightforward, because our pure formulas are conjunctions rather than arbitrary Boolean combinations of atomic Presburger formulas.

Nevertheless, we can detect the
essential difference between the biabduction and entailment
problems for $\ASL$ (at least in the case where the existential quantifiers in $B$ are restricted as described above).  Namely, by Theorem~\ref{t-entail-P2-hard} entailment is still \mbox{$\Pi_2^P$-}hard
whereas, by Props.~\ref{prop:runtime_NP} and~\ref{prop:biabd_qf}, the biabduction problem belongs to~$\NP$.  

\section{Related work}
\label{sec:related}

The literature most closely related to our work in the present paper divides, broadly speaking, into four main categories.

\paragraph{Separation logic over linked list segments.}

Perhaps the most popular and extensively studied part of separation logic is the symbolic heap fragment over linked lists, introduced and shown decidable in~\cite{Berdine-Calcagno-OHearn:04}.  This fragment is essentially the one employed in Facebook's \tool{Infer} tool~\cite{Calcagno-etal:15}.  Here, the pure part of symbolic heaps is a conjunction of simple equalities and disequalities between expressions (typically just variables or the constant $\nil$), while the spatial part admits points-to formulas $E \mapsto E'$, denoting a single pointer in the heap, and list segment formulas of the form $\ls{E}{E'}$, denoting a linked list in the heap from $E$ to $E'$.

Following the initial decidability result, it was shown in~\cite{Cook-etal:11} that satisfiability and entailment in this logic are in $\PTIME$.  The biabduction problem for this fragment and practical approaches to it were first studied in~\cite{Calcagno-etal:11}; in~\cite{Gorogiannis-etal:11} it was shown that the \emph{abduction} problem (where only an ``antiframe'' $X$ is computed) is in fact $\NP$-complete.

We observe that this fragment and our $\ASL$ are largely disjoint: our $\mathsf{array}$s cannot be defined in terms of $\mathsf{ls}$s, or vice versa, while $\ASL$ also employs arithmetic formulas rather than simple (dis)equality constraints. This is also reflected in the differences in their respective complexity bounds.

\paragraph{Separation logic with inductive predicates.}

There has been substantial research interest in (symbolic heap) separation logic over general \emph{inductively defined predicates}~\cite{Brotherston:07}, as opposed to fixed data structures such as lists (or indeed arrays).  Such predicates can be used to describe arbitrary data structures in memory; they might be provided to an analysis by the user, or perhaps inferred automatically (cf.~\cite{Brotherston-Gorogiannis:14,Le-etal:14}).

When arbitrary inductive definitions over symbolic heaps are permitted, the entailment problem is undecidable~\cite{Antonopoulos-etal:14} while satisfiability and even model checking (i.e., deciding whether a given stack-heap pair satisfies a given formula) become $\EXPTIME$-complete (cf.~\cite{Brotherston-etal:14} resp.~\cite{Brotherston-etal:16}).  More tractable fragments can be obtained by restricting the admissible forms of inductive definitions. A fragment in which all definitions have \emph{bounded treewidth}~\cite{Iosif-etal:13} was shown to have a decidable entailment problem by reduction to bounded-treewidth monadic second-order logic; a variant of this fragment, with different restrictions, was similarly shown decidable in~\cite{Tatsuta-Kimura:15}. However, our $\ASL$ cannot be encoded even in the unrestricted fragment, owing to the absence of arithmetic.

Very recently, in~\cite{Gu-Chen-Wu:16}, decidability of satisfiability and entailment was obtained for a fragment of symbolic-heap separation logic with restricted inductive predicates (called ``linearly compositional'') \emph{and} Presburger arithmetic constraints. However, $\ASL$ cannot be encoded in this fragment, because pointers and data variables belong to disjoint sorts, effectively disallowing pointer arithmetic. Moreover, we provide an analysis of biabduction, which is the central focus of our paper, but not considered in~\cite{Gu-Chen-Wu:16}.

Finally, also very recently, a semidecision procedure for satisfiability in symbolic-heap separation logic with inductive definitions and Presburger arithmetic appeared in~\cite{Le-Sun-Chin:16}.  $\ASL$ can be encoded in their logic, but, as far as we can tell, not into the subfragment for which they show satisfiability decidable.  We note that in any case this decidability result comes without any complexity bounds.

\paragraph{Separation logic with iterated separating conjunction.}
The \emph{iterated separating conjunction} (ISC)~\cite{Reynolds:02}, a binding operator for expressing various
unbounded data structures, was recognised early on as a way of reasoning about arrays. For example, the ISC was employed recently in a framework for reasoning about memory permissions, with the aim of enabling symbolic execution of
concurrent array-manipulating programs~\cite{Muller-etal:16}. An earlier paper employing a form of ISC and biabduction is \cite{Gulavani:09}, where the aim is to design a bottom-up shape analysis 
for unannotated code.

However, although our $\textsf{array}$ predicate can be expressed using the ISC, we do not know of any existing decision procedures for biabduction, entailment or even satisfiability in such a logic, which may be of higher complexity or become undecidable (there is certainly no investigation of these issues in either \cite{Gulavani:09} or \cite{Muller-etal:16}). Our work is aimed at underpinning compositional analyses of unannotated code; in contrast, the analysis promoted in \cite{Muller-etal:16} requires fully annotated programs and does not employ, or investigate, biabduction. As for~\cite{Gulavani:09}, arrays are not considered and arithmetic is disallowed (even though arrays are expressible with its ISC); therefore array-manipulating programs cannot be treated.

\paragraph{Other program analyses on arrays.}
A significant amount of research effort has previously focused on the verification of
array-manipulating programs either via invariant inference and theorem proving,
or via abstract interpretation (for instance \cite{Kovacs:09,Dillig-et-al:10,Cousot-et-al:11,Bouajjani:2012,Alberti:14,Stroder-etal:16}).
These approaches differ from ours technically, but also in intention. First, the emphasis in these investigations is on data constraints and, thus, tends towards proving general safety properties of programs.
Here, we intentionally restrict the language so that we can obtain sound and complete algorithms
which can be used for establishing memory safety of programs but not for
proving arbitrary safety properties. Second, such approaches are typically
whole-program analyses that cannot be used in a bottom-up fashion or on partial programs.
In contrast, our focus is on biabduction, one of the key ingredients that makes such
a compositional approach possible.

\section{Conclusions and future work}
\label{sec:conclusion}

In this paper, we investigate $\ASL$, a separation logic aimed at compositional memory safety proofs for array-manipulating programs. We focus on \emph{biabduction}, the key to interprocedural specification inference: we give a sound and complete $\NP$ algorithm for biabduction that computes solutions by finding a consistent ordering of the array endpoints, and we show that the problem is $\NP$-hard in the quantifier-free case. In addition, we show that the satisfiability problem for $\ASL$ is $\NP$-complete, and entailment is decidable, being $\CoNP$-complete for quantifier-free formulas, and at least $\Pi^P_2$-hard (perhaps much harder) in general. We believe that ours are the first decision procedures for separation logic over arrays; certainly, we believe that we are the first to treat biabduction in this context.

The obvious direction of travel for future work is to build an abductive program analysis {\`a la} \textsc{Infer}~\cite{Calcagno-etal:15} for array programs, using $\ASL$ as the assertion language. The first step is to implement an algorithm for biabduction. A direct implementation of our algorithm in Section~\ref{sec:qfree-biabduction}, using an SMT solver to find a solution seed, is the most immediate possibility, but not the only one; one might also try possibly-incomplete but fast approaches based on theorem proving (cf.~\cite{Calcagno-etal:11}). A currently extant problem is in finding biabduction solutions that are as logically weak as possible; our algorithm currently commits to a total ordering of all arrays even if a partial ordering would be sufficient.  We believe that, in practice, this could be resolved by refining the notion of a solution seed so that it carries \emph{just} enough information for computing the spatial formulas in $X$ and $Y$. A more conceptually interesting problem is how we might assess the quality of logically incomparable biabduction solutions (e.g. according to the amount of memory they occupy).

In addition, a program analysis for $\ASL$ will rely not just on biabduction but also on suitable \emph{abstraction} heuristics for discovering loop invariants; this seems an interesting and non-trivial problem for the near future.

Finally, readers might wonder about the possibility of combining $\ASL$ with other fragments of separation logic, such as the linked list fragment, for expressivity reasons. Certainly, we expect that some programs might manipulate, e.g., both linked lists \emph{and} arrays at the same time (and possibly other dynamic data structures too), and a combined language would then clearly be needed to reason about such programs.  However, it is not clear whether such a logic (with, say, arithmetic constraints, arrays and linked lists) would enjoy good computational properties; a potentially problematic issue is that a heap might simultaneously satisfy, e.g., a \mbox{$*$-conjunction} of single heap cells, an array \emph{and} a linked list, all at the same time. We consider this a very interesting area for future study.

\bibliographystyle{abbrvnat}
\bibliography{array_biab_refs}

\clearpage
\appendix

\newtheorem{innerapplemma}{Lemma}
\newenvironment{applemma}[1]
  {\renewcommand\theinnerapplemma{#1}\innerapplemma}
  {\endinnerapplemma}

\newtheorem{innerappprop}{Proposition}
\newenvironment{appprop}[1]
  {\renewcommand\theinnerappprop{#1}\innerappprop}
  {\endinnerappprop}

\newtheorem{innerappthm}{Theorem}
\newenvironment{appthm}[1]
  {\renewcommand\theinnerappthm{#1}\innerappthm}
  {\endinnerappthm}

\section{Proofs of results in Section~\ref{sec:language}}

\begin{applemma}{\ref{lem:semantics}}
For all quantifier-free symbolic heaps $A$, if \mbox{$s,h \models A$} and $s,h' \models A$, then $\dom{h} = \dom{h'}$.
\end{applemma}

\begin{proof}
Writing $A = \Pi : F$, we proceed by structural induction on the spatial part $F$.

\proofcase{$F = \emp$} By definition, $\dom{h} = \dom{h'} = \emptyset$.

\proofcase{$F = \psto{t_1}{t_2}$} By definition, $\dom{h} = \dom{h'} = \{s(t_1)\}$.

\proofcase{$F = \absarray{t_1}{t_2}$} By definition, $\dom{h} = \dom{h'} = \{s(t_1),\ldots,s(t_2)\}$.

\proofcase{$F = F_1 * F_2$}. We have $h = h_1 \circ h_2$ and $h' = h_1' \circ h_2'$, where $s,h_1 \models F_1$ and $s,h_1' \models F_1$, and $s,h_2 \models F_2$ and $s,h_2' \models F_2$. Since $s,h_1 \models F_1$ and $s,h_1' \models F_1$, we have $\dom{h_1} = \dom{h_1'}$ by induction hypothesis.  Similarly, $\dom{h_2} = \dom{h_2'}$.  Because $\circ$ is defined as the union of domain-disjoint heaps, it follows that $\dom{h_1 \circ h_2} = \dom{h_2 \circ h_2'}$.  That is, $\dom{h} = \dom{h'}$ as required.  This completes the induction.
\end{proof}

\section{Proofs of results in Section~\ref{sec:satisfiability}}

\begin{applemma}{\ref{lem:abstraction}}
Let $A$ be a quantifier-free symbolic heap and $s$ a stack.
Then,
\[
\exists h.\ s,h\models A \;\iff\; \exists h'.\ s,h'\models \abstr{A}.
\]
\end{applemma}

\begin{proof}
Let $A$ and $\abstr{A}$ be as shown in Defn.~\ref{defn:abstraction}. \medskip

\noindent{($\Rightarrow$)}
Immediate by the semantics of $*$ and the observation that
$\psto{c_j}{d_j}\models \absarray{c_j}{c_j}$ for all $j\in[1,m]$. \medskip

\noindent{($\Leftarrow$)}
Let $s,h$ be a model of $\abstr{A}$.
We define a model $s,\hat h$ such that $s,\hat h \models A$. First, by assumption
we have $s\models\Pi$. Also, there exist disjoint heaps $h_1,\ldots,h_n,h'_1,\ldots,h'_m$ such that
$h=h_1\circ\ldots\circ h_n\circ h'_1\circ \ldots\circ h'_m$ and
$s,h_i\models\absarray{a_i}{b_i}$ for $i\in[1,n]$, and
$s,h'_j\models\absarray{c_j}{c_j}$ for $j\in[1,m]$.
We define new heaps $h''_1,\ldots,h''_m$ as follows. The heap $h''_j$ is
defined by $\dom{h''_j} = \{s(c_j)\}$ and
$h''_j(s(c_j)) = s(d_j)$ for all $j\in[1,m]$. We then define a new heap
$\hat h = h_1\circ\ldots\circ h_n\circ h''_1\circ\ldots\circ h''_m$, which is
well defined by the fact that $\dom{h''_j}=\dom{h'_j}$ and the assumption that
$h$ is well defined. It is easy to see that $s,h''_j\models \psto{c_j}{d_j}$
and by the semantics of $*$ we are done.
\end{proof}

\begin{applemma}{\ref{lem:gamma}}
For any stack $s$ and any quantifier-free 
symbolic heap $A$,
\[
s \modelsPA \gamma(A)\;\iff\;  \exists h.\ s,h \models A.
\]
\end{applemma}

\begin{proof}
First, note that satisfiability of $A$ coincides with the satisfiability of $\abstr{A}$ by Lemma~\ref{lem:abstraction}.  Thus it suffices to consider the case when $A$ is $\mapsto$-free.  
We assume then that $A = \Pi : \bigsepstar^n_{i=1}\absarray{a_i}{b_i}$,
and establish each direction of the lemma separately. \\

\noindent{($\Leftarrow$)}
Suppose that $s,h \models A$. That is, $s \models \Pi$ and there exist heaps $h_1,\ldots,h_n$ such that $h = h_1 \circ \ldots \circ h_n$ and $s,h_i \models \absarray{a_i}{b_i}$ for each $i\in[1,n]$. We require to show that $s \modelsPA \gamma(A)$.

First, $s \modelsPA \Pi$ by assumption. Next, for each \mbox{$i\in[1,n]$}, we have $s(a_i) \leq s(b_i)$ because $s,h_i \models \absarray{a_i}{b_i}$; hence $s \modelsPA \textstyle\bigwedge_{1 \leq i \leq n}a_i \leq b_i$.  Finally, letting $1 \leq i < j \leq n$, we have $\dom{h_i} = \{s(a_i),\ldots,s(b_i)\}$ and $\dom{h_j} = \{s(a_j),\ldots,s(b_j)\}$.  Since $\dom{h_i}$ and $\dom{h_j}$ are disjoint by assumption, we must have either $s(b_i) < s(a_j)$ or $s(b_j) < s(a_i)$, therefore
\mbox{$s \modelsPA \textstyle\bigwedge_{1 \leq i<j \leq n}(b_i < a_j) \vee (b_j < a_i)$}.
Putting everything together, $s \modelsPA \gamma(A)$ as required. \\

\noindent{($\Rightarrow$)}
Supposing that $s \modelsPA \gamma(A)$, of the form above, we require to construct a heap $h$ such that $s,h \models A$. For each $i\in[1,n]$, define a heap $h_i$ by $\dom{h_i} = \{s(a_i),\ldots,s(b_i)\}$ (the contents may be chosen arbitrarily). We observe that $\dom{h_i}$ is well defined because $s \modelsPA a_i \leq b_i$ by assumption. By construction, we have $s,h_i \models \absarray{a_i}{b_i}$.

Next, we claim that $h_1 \circ \ldots \circ h_n$ is defined.  Supposing not (for contradiction), then there exist $h_i,h_j$ with $1 \leq i < j \leq n$ such that $\dom{h_i} \cap \dom{h_j} \neq \emptyset$.  That is,  $\{s(a_i),\ldots,s(b_i)\} \cap \{s(a_j),\ldots,s(b_j)\} \neq \emptyset$, which means that (without loss of generality) either $s(a_j)$ or $s(b_j)$ falls within the range $\{s(a_i),\ldots,s(b_i)\}$; i.e., either $s(a_i) \leq s(a_j) \leq s(b_i)$ or $s(a_i) \leq s(b_j) \leq s(b_i)$ (or both).  By assumption, we have $s \modelsPA (b_i < a_j) \vee (b_j < a_i)$, meaning that either $s(b_i) < s(a_j)$ or $s(b_j) < s(a_i)$. This gives us four cases to consider, and it is simple to see that each leads to a contradiction:
{\bf (1)}
if $s(a_i) \leq s(a_j) \leq s(b_i)$ and $s(b_i) < s(a_j)$, we immediately get $s(a_j) < s(a_j)$;
{\bf (2)} if $s(a_i) \leq s(a_j) \leq s(b_i)$ and $s(b_j) < s(a_i)$, we get $s(b_j) < s(a_j)$, contradicting $s \modelsPA a_j \leq b_j$;
{\bf (3)}
if $s(a_i) \leq s(b_j) \leq s(b_i)$ and $s(b_i) < s(a_j)$, we again get $s(b_j) < s(a_j)$;
{\bf (4)}
if $s(a_i) \leq s(b_j) \leq s(b_i)$ and $s(b_j) < s(a_i)$, we get $s(a_i) < s(a_i)$.
Putting everything together, and using the fact that $s \models \Pi$, we obtain $s,h_1 \circ \ldots \circ h_n \models A$, and are done.
\end{proof}

\begin{applemma}{\ref{lem:3part_to_sat}}
Given a 3-partition problem instance $(B,\cal{S})$, we have
\[
A_{B,\cal{S}} \mbox{ is satisfiable } \;\iff\;
 \exists \mbox{ complete 3-partition of } \cal{S} \mbox{ (w.r.t. $B$),}
\]
where $A_{B,\cal{S}}$ is the symbolic heap given by Definition~\ref{defn:3part_to_sat}.
\end{applemma}

\begin{proof}
We establish each direction of the equivalence separately. \\

\noindent($\Leftarrow$)\
 Let \mbox{$\{(k_{j_{i,1}}, k_{j_{i,2}}, k_{j_{i,3}})
  \mid 1 \leq i \leq m\}$} be a complete \mbox{$3$-}partition
 of\/~$\cal{S}$.
 We define a stack $s$ by $s(d_1) = 0$ and, for all \mbox{$1 \leq i \leq m$},
\[\begin{array}{rcl}              
s(d_{i+1}) & = & s(d_i) + B+1,
\\
s(a_{j_{i,1}}) & = & s(d_i),
\\
s(a_{j_{i,2}}) & = & s(a_{j_{i,1}}) + k_{j_{i,1}},
\\
\mbox{ and } s(a_{j_{i,3}}) & = & s(a_{j_{i,2}}) + k_{j_{i,2}}\ .
\end{array}\]
Notice that, using the equation $k_{j_{i,1}} + k_{j_{i,2}} + k_{j_{i,3}} = B$, we have
\[
s(a_{j_{i,3}}) + k_{j_{i,3}} + 1 = s(d_i) + B + 1 = s(d_{i+1})\ .
\]%
\noindent
Next we define a heap $h$ (with arbitrarily chosen contents) by
\[
\dom{h} = \{s(d_1),s(d_1)+1,\ldots,s(d_{m+1})\}\ .
\]%
We claim that \mbox{$s,h \models A_{B,\cal{S}}$},
 as defined above. 

First, we tackle the pure part. First, for each $1 \leq i \leq m$, we have $s \models d_{i+1} = d_i + B + 1$ by definition. Next, for each $1 \leq j \leq 3m$, we have by construction $a_j \geq 0 = d_1$.  Finally, for all $1 \leq j \leq 3m$ we have, by construction and using the assumed bounds on each $k$,
\[\begin{array}{rcl}
s(a_j) \leq s(a_{j_{m,3}})& = & s(d_m) + k_{j_{i,1}} + k_{j_{i,2}} \\
& \leq & s(d_i) + B/2 + B/2 \\
& < & s(d_i) + B + 1 \\
& = & s(d_{m+1})\ .
\end{array}\]%
Thus indeed $s$ satisfies the pure part of $A_{B,\cal{S}}$.

Next, we check that $s,h$ models the spatial part.  We define $m+1$ ``heaplets'' $h_{d_i}$ by $\dom{h_{d_i}} = \{s(d_i)\}$ for each $1 \leq i \leq m$, and $3m$ heaplets $h_{j_{i,\ell}}$ for each $1 \leq i \leq m$ and $\ell \in \{1,2,3\}$ by
\[\begin{array}{rcl}
\dom{h_{j_{i,1}}} & = & \{s(d_i)+1,\ldots,s(a_{j_{i,2}})\} \\
\dom{h_{j_{i,2}}} & = & \{s(a_{j_{i,2}})+1,\ldots,s(a_{j_{i,3}})\} \\
\dom{h_{j_{i,3}}} & = & \{s(a_{j_{i,3}})+1,\ldots,s(d_{i+1})-1\}
\end{array}\]%
(As before, the contents of these heaplets are irrelevant.)
\\ By construction $s,h_{d_i} \models \relarray{d_i}{0}{0}$ for each
 \mbox{$1 \leq i\leq m+1$}.
 Similarly, for each \mbox{$1 \leq i \leq m$} and
\mbox{$\ell \in \{1,2,3\}$} we have that
 \mbox{$s,h_{j_{i,\ell}} \models
 \relarray{a_{j_{i,\ell}}}{1}{k_{j_{i,\ell}}}$}.
  Since each $j_{i,\ell}$
corresponds to a unique element in the sequence $\cal{S}$,
 this gives us the following
\mbox{$s,h_j \models \relarray{a_j}{1}{k_j}$}
 for each $1 \leq j \leq 3m$.  We
define $h$ to be the $\circ$-composition of all our heaplets, i.e.,

\[
h = \bigsepcirc_{1 \leq i \leq m+1} h_{d_i}
 \circ \bigsepcirc_{1 \leq j \leq 3m} h_j \ ,
\]%
where the indexed ``big circle'' notation abbreviates a $\circ$-composition of heaps.
To see that $s,h \models A_{B,\cal{S}}$, we just need to show that $h$ is well-defined, i.e., that all of our heaplets are non-overlapping.  This holds by construction: for any $1 \leq i \leq m$ we have that $h_{d_i}$ and $h_{d_{i+1}}$ are single cells separated by a contiguous gap of $B$ cells, and the heaplets $h_{j_{i,1}}$, $h_{j_{i,2}}$ and $h_{j_{i,3}}$ are disjoint heaps occupying the gap between $h_{d_i}$ and $h_{d_{i+1}}$.  Thus $s,h \models A_{B,\cal{S}}$ as required. \\

\noindent($\Rightarrow$) Let $s,h$ be a stack-heap pair satisfying $s,h \models A_{B,\cal{S}}$.  The spatial part of $A_{B,\cal{S}}$ immediately yields a decomposition of $h$ as
\[
h = \bigsepcirc_{1 \leq i \leq m+1} h_{d_i} \circ
 \bigsepcirc_{1 \leq j \leq 3m} h_j \ ,
\]%
where $\dom{h_{d_i}} = \{s(d_i)\}$ for each $1 \leq i \leq m+1$ and
$\dom{h_j} = \{s(a_j)+1,\ldots,s(a_j)+k_j\}$ for each $1 \leq j \leq
3m$; moreover, all of these ``heaplets'' are non-overlapping.  In
addition, the spatial part of $A_{B,\cal{S}}$ yields $s(d_{i+1}) =
s(d_i) + B + 1$ for all $1 \leq i \leq m+1$, plus $s(d_1) \leq s(a_j)$
and $s(a_j) + k_j \leq d_{m+1}$ for all $1 \leq j \leq 3m$.  This
immediately implies that each heaplet $h_j$ occupies a contiguous block
of $k_j$ cells between two successive single-cell heaplets $h_{d_i}$ and
$h_{d_{i+1}}$, which are themselves separated by a block of $B$ cells.
Moreover, because of the above equation
\mbox{$\Sigma_{j=1}^{3m}k_j = mB$},
 every such block of $B$ cells must be
\emph{exactly covered} by $h_j$ heaplets.

Now, we observe that, for each $i$, the block of $B$ cells between
$h_{d_i}$ and $h_{d_{i+1}}$ must be covered by \emph{precisely three} of
our $3m$ heaplets: $h_{j_{i,1}}$, $h_{j_{i,2}}$ and $h_{j_{i,3}}$, say.
This is due to the fact that $B/4 < k_j < B/2$ for each $j$: two
heaplets are therefore insufficient to fill a gap of $B$ cells, whereas
four heaplets would occupy more than $B$ cells (and would therefore
overlap with each other or with $h_{d_i}$ or $h_{d_{i+1}}$).

Therefore, we can define a 3-partition of $\cal{S}$ by taking for each $1 \leq i \leq m$ the numbers $k_{i,1}$, $k_{i,2}$ and $k_{i,3}$ given by the sizes of the heaplets occupying the cells between $h_{d_i}$ and $h_{d_{i+1}}$. It is immediate that $k_{i,1} + k_{i,2} + k_{i,3} = B$, as required.
\end{proof}

\section{Proofs of results in Section~\ref{sec:biabduction}}
\label{s-appendix-biabduct} 

\begin{appprop}{\ref{prop:solution-to-beta-sat}}
If the biabduction problem $(A,B)$ has a solution, then $\beta(A,B)$ is satisfiable.
\end{appprop}

\begin{proof}
Let $X,Y$ be a solution for $(A,B)$. This means that
$A*X$ is satisfiable and that $A*X\models B*Y$. We conclude there exists a model
that there is a model $s,h$ such that $s,h\models A*X$ \emph{and} $s,h\models B*Y$.

Since $s,h\models A*X$ this means that there is a subheap $h'\subseteq h$ such
that $s,h'\models A$. Applying Lemma\nobreakspace \ref {lem:gamma} to $s,h'$, we obtain that
$s\models\gamma(A)$. The same reasoning on $s,h\models B*Y$ yields $s\models\gamma(B)$.
It remains to show that
\[
s\models\bigwedge_{j=1}^\ell\bigwedge_{i=1}^n (v_j<a_i \lor v_j>b_i) \land
\bigwedge_{i=1}^\ell\bigwedge_{j=1}^k (t_i\neq v_j\lor u_i=w_j )\ .
\]

Suppose the left conjunct is false. Then, there are \mbox{$j\in[1,\ell]$} and $i\in[1,n]$
for which $s\models a_i\le v_j \le b_i$. This means that the heap $h_j = \cutheap{\psto{v_j}{w_j}}$
is a \emph{subheap} of the heap $h_i = \cutheap{\absarray{a_i}{b_i}}$. Let $\xi\in\val$ such that
$\xi\neq s(w_j)$. It is easy to see that $s,h_i[s(v_j)\mapsto\xi]\models\absarray{a_i}{b_i}$
because the array predicate is insensitive to the values stored in the heap.
This also means that $s,h[s(v_j)\mapsto\xi]\models A*X$. At the same time it is
clear that $s,h_j[s(v_j)\mapsto\xi]\not\models \psto{v_j}{w_j}$. Therefore
$s,h_j[s(v_j)\mapsto\xi]\not\models B*Y$, contradiction.

Suppose the right conjunct is false. Then, there are \mbox{$i\in[1,\ell]$} and $j\in[1,k]$
such that $s\models t_i=v_j\land u_i\neq w_j$. Thus the heap $h_i=\cutheap{\psto{t_i}{u_i}}$
is well-defined, since $s,h\models A*X$.
Similarly, the heap $h_j = \cutheap{\psto{v_j}{w_j}}$ is well-defined,
because $s,h\models B*Y$. However, $s\models t_i=v_j$ meaning that $\dom{h_i}=\dom{h_j}$.
On the other hand, $h_i(s(u_i))\neq h_j(s(w_j))$ since $s\models u_i\neq w_j$. This
is a contradiction, because both $h_i$ and $h_j$ are subheaps of $h$, but they have the
same domain. This completes the proof.
\end{proof}

\begin{appthm}{\ref{thm:beta-sat-implies-seed}}
If $\beta(A,B)$ is satisfiable,
then there exists a solution seed $\Delta$ for the biabduction problem $(A,B)$.
\end{appthm}
\begin{proof}
Supposing $s \modelsPA \beta(A,B)$, we define $\Delta$ as follows:
\[
\Delta \defeq
\bigwedge_{\substack{e,f \in \abdterms{A}{B} \\ s(e) < s(f)}} e < f
\;\;\wedge\!
\bigwedge_{\substack{e,f \in \abdterms{A}{B} \\ s(e) = s(f)}} e = f .
\]%
We now check that $\Delta$ satisfies the conditions in \MakeUppercase defn.\nobreakspace \ref {defn:seed}.

\MakeUppercase condition\nobreakspace \ref {item:order} holds because $\le$ is a total order over the set
\mbox{$\{s(e) \mid e\in\abdterms{A}{B}\}$}.
Thus, the definition of $\Delta$ will introduce one of the atoms $f < e$, $e < f$ or $f =e$, for
all $e,f\in\abdterms{A}{B}$.

\MakeUppercase condition\nobreakspace \ref {item:terms} holds by construction.

\MakeUppercase condition\nobreakspace \ref {item:delta} requires that $\Delta$ is satisfiable.
This follows by construction, as clearly $s$ is a model of $\Delta$.

\MakeUppercase condition\nobreakspace \ref {item:delta} also requires that $\Delta \modelsPA \beta(A,B)$. First we show $\Delta \modelsPA \gamma(A)$.
Recall (\MakeUppercase defn.\nobreakspace \ref {defn:gamma}) that, supposing $A$ is written as in \MakeUppercase defn.\nobreakspace \ref {defn:seed}, we have
\[
\gamma(A) = \Pi\ \wedge \!\!\bigwedge_{i\in[1,n+k]} \!\!\! \hat{a}_i \leq \hat{b}_i\ \wedge \!\!
\bigwedge_{1 \leq i < j \leq n+k} \!\!\!(\hat{b}_i < \hat{a}_j) \vee (\hat{b}_j < \hat{a}_i)
\]%
where $\hat{a}_i,\hat{b}_i$ are the endpoints of arrays in $\abstr{A}$, of which there are
exactly $n+k$.
Suppose $\pi$ is a conjunct in $\Pi$. If $\pi$ is of the form $t = u$ then,
since $s \modelsPA \gamma(A)$ and thus $s \models \Pi$, we have $s \models t = u$;
therefore by construction the conjunct $(t = u)$
appears in $\Delta$ and thus trivially $\Delta \modelsPA t = u$. The case for $t < u$ is similar.
Suppose then that $\pi$ is of the form $t \le u$. Then, either $s(t) = s(u)$,
in which case $(t = u)$ appears in $\Delta$, or $s(t) < s(u)$ in which case
$(t < u)$ appears in $\Delta$.  In both cases, $\Delta \modelsPA t \le u$.
Finally, if $\pi$ is $t \neq u$ then it must be the case that
either $(t < u)$ or $(u < t)$ appears in $\Delta$, which again means that
$\Delta \modelsPA t \neq u$. Therefore $\Delta \modelsPA \Pi$.

Next, let $i\in[1,n+k]$, and observe $\hat{a}_i,\hat{b}_i \in \abdterms{A}{B}$.
Since $s \modelsPA \gamma(A)$, we have $s(\hat{a}_i) \leq s(\hat{b}_i)$,
meaning that either $s(\hat{a}_i) < s(\hat{b}_i)$ or $s(\hat{a}_i) = s(\hat{b}_i)$.
Thus, by construction, either $(\hat{a}_i < \hat{b}_i)$ or $(\hat{a}_i = \hat{b}_i)$
is a conjunct of $\Delta$, and in both cases $\Delta \modelsPA \hat{a}_i \leq \hat{b}_i$.
Therefore, \mbox{$\Delta \modelsPA \bigwedge_{i\in[1,n+k]} \hat{a}_i\le \hat{b}_i$}.

Finally, let $1 \leq i < j \leq n$, and observe $\hat{a}_i,\hat{b}_i,\hat{a}_j,\hat{b}_j$
are all in $\abdterms{A}{B}$.  Since $s \modelsPA \gamma(A)$ by assumption, we
have $s \modelsPA (\hat{b}_i < \hat{a}_j) \vee (\hat{b}_j < \hat{a}_i)$,
meaning that either $s(\hat{b}_i) < s(\hat{a}_j)$ or $s(\hat{b}_j) < s(\hat{a}_i)$.
Thus either $(\hat{b}_i < \hat{a}_j)$ or $(\hat{b}_j < \hat{a}_i)$ is a conjunct of
$\Delta$, so $\Delta \modelsPA (\hat{b}_i < \hat{a}_j) \vee (\hat{b}_j < \hat{a}_i)$.
This gives us
$\Delta \modelsPA \bigwedge_{1 \leq i < j \leq n} (\hat{b}_i < \hat{a}_j) \vee (\hat{b}_j < \hat{a}_i)$.
Putting everything together, we get $\Delta \modelsPA \gamma(A)$.
The argument that $\Delta \modelsPA \gamma(B)$ is identical.

Next, we show
$\Delta\models \bigwedge_{j=1}^\ell\bigwedge_{i=1}^n (v_j<a_i \lor v_j>b_i)$.
We know that $s\models v_j<a_i \lor v_j>b_i$ for all $j\in[1,\ell]$ and $i\in[1,n]$.
Thus $s\models v_j<a_i$ or $s\models v_j>b_i$, meaning $s(v_j)<s(a_i)$ or $s(v_j)>s(b_i)$.
By the fact $v_j,a_i,b_i\in\abdterms{A}{B}$ and the definition of $\Delta$ we know that
one of $(v_j<a_i)$ or $(b_i<v_j)$ is a conjunct of $\Delta$. Thus $\Delta\models (v_j<a_i)\lor (b_i<v_j)$
and we are done.

Finally, we show $\Delta\models \bigwedge_{i=1}^\ell\bigwedge_{j=1}^k (t_i\neq v_j\lor u_i=w_j )$.
Again, we know that $s\models t_i\neq v_j\lor u_i=w_j$ for all $i\in[1,\ell]$ and $j\in[1,k]$.
There are two cases: $s\models u_i=w_j$ or $s\models  t_i\neq v_j$.  In the first case,
$\Delta\models u_i=w_j$ by construction. In the latter case, there are two further subcases,
namely $s\models t_i<v_j$ or $s\models t_i>v_j$ and it can be easily seen that in both
of these, $\Delta\models t_i\neq v_j$.
This completes the proof.
\end{proof}

\begin{appprop}{\ref{prop:termination}}
Let $(A,B)$ be a biabduction problem of the form shown in Defn.~\ref{defn:seed}.
Let $\Delta$ be a solution seed and terms $e,f \in \abdterms{A}{B}$.
The call $\arrcov{A,\Delta}{e}{f}$:
\begin{enumerate}
\item always terminates, issuing up to $n+k$ recursive calls;
\item returns a formula (for some $q\in\nat$ and sets $I,J\subseteq\nat$)
\[\bigwedge_{i\in I} a_i = a_i'+1 \land
\bigwedge_{i\in J} t_i = t_i'+1  : \bigsepstar_{i=1}^q \absarray{l_i}{r_i}\]
where for all $i\in[1,q]$, $l_i \in \abdterms{A}{B}$;
\item for every $i\in[1,q]$, $\hat \Delta \modelsPA e \le l_i \le r_i \le f$;
\item for every $i\in[1,q-1]$, $\hat \Delta \models r_i<l_{i+1}$.
\end{enumerate}
\end{appprop}

\begin{proof}
First, note that there are exactly $n+k$ arrays in $\abstr{A}$,
hence the upper limit of $\bigsepstar$ in line~\ref{ln:let-abstr}.

Termination follows from the fact that $\arrcov{A,\Delta}{e}{f}$ either terminates
immediately when \mbox{$f \Dlt e$}, or recurses with calls of the form
\mbox{$\arrcov{A,\Delta}{b_{i_j}+1}{f}$}, where the sequence $b_{i_j}$ is $\Dlt$-increasing,
thus terminating at the first index $i_j$ such that $f \Dlt b_{i_j}+1$.
There can be up to $n+k$ such calls.

To show items 2 and 3, we examine each section of the algorithm,
and argue by induction over the recursion depth.

If $f \Dlt e$ then the algorithm terminates at line~\ref{ln:return-emp},
returning $\emp$, a result of the required form.

Otherwise, $e \Dleq f$ (by Lemma~\ref{lem:delta-ordering}).
If $a_i \Dleq e \Dleq b_i$ for some $i\in[1,n+k]$ (line~\ref{ln:t-covered}), then the recursive call
$\arrcov{A,\Delta}{b_i+1}{f}$ is issued. Since $e \Dleq b_i$, we know that
$e \Dlt b_i+1$.

Otherwise, there is no $i$ such that $a_i \Dleq e \Dleq b_i$.
If the set $E$ is empty (line~\ref{ln:e-empty}), then the algorithm terminates
returning a result that is, trivially, of the required form.

Otherwise, there is a minimal element in $E$, namely $a_i$.
In this case, a recursive call $\arrcov{A,\Delta}{b_i+1}{f}$ is issued, with $e < b_i+1$.
By the inductive hypothesis and the lifting of~$*$ to symbolic heaps, we
obtain a result of the required form.

That for every $i\in[1,q]$, $\hat \Delta \modelsPA e \le l_i \le r_i \le f$
follows by inspecting the array constructors used in the code.
In particular, $\absarray{e}{f}$ (line~\ref{ln:array-t-u}) trivially provides the
required condition (note that $e \Dleq f$ by line~\ref{ln:u-less-than-t}).
For $\absarray{e}{\hat{a}'_i}$ at line~\ref{ln:array-t-ai}, observe that
$e \Dlt \hat{a}_i$ holds by the definition of $E$ at line~\ref{ln:def-e}.
Moreover, \mbox{$\Theta_X\modelsPA \hat{a}_i = \hat{a}'_i + 1$}, thus
$\hat\Delta\modelsPA e \le \hat{a}'_i$.

Line~\ref{ln:array-t-ai} also guarantees item (4): this is the only place in the code
where multiple arrays may be returned, and we clearly have
$\hat\Delta\models\hat{a}'_i<\hat{b}_i+1$, which, combined with item~(3) completes the proof.
\end{proof}

We will use the expression $\heaplet{\nu}{\xi}$, where $\nu,\xi\in\nat$,
to denote the heap~$h$ such that $\dom{h}=\{\nu\}$ and $h(\nu)=\xi$.

\begin{applemma}{\ref{lem:alg_sat}}
Let $(A,B)$ be a biabduction instance, $\Delta$ a solution seed
and $X$ as in Defn.~\ref{alg:biabduction}.
Then, $A*X$ is satisfiable.
\end{applemma}
\begin{proof}
By \MakeUppercase defn.\nobreakspace \ref {defn:seed} we know there is a stack $\hat s$ such that $\hat s\models\Delta$.
We define a stack $s_X$ that correctly assignes values to primed terms, as added
by $\arrcovalg$.
\[
s_X(e) \defeq
\begin{cases}
\hat s(e) & e \in \abdterms{A}{B} \\
\hat s(a_i) - 1 & e\equiv a'_i\in\fv{X}, \text{ for $i\in[1,n]$} \\
\hat s(t_i) - 1 & e\equiv t'_i\in\fv{X}, \text{ for $i\in[1,k]$}
\end{cases}
\]
Observe that the variables $a'_i$ and $t'_i$ are fresh in $\Delta$ and
appear at most once in $\Theta_X$ (this is due to \MakeUppercase prop.\nobreakspace \ref {prop:termination}).
We must show that $s_X$ is well defined, i.e., there is no term $a'_i\in\fv{X}$
such that $\hat s(a_i)=0$, and equally for terms $t'_i$.

Suppose there is such an $a'_i$. Then there must be some $j\in[1,n]$ for which the
call $\arrcov{A,\Delta}{a_j}{b_j}$ reaches line~\ref{ln:array-t-ai} which
introduces the constraint $a'_i+1=a_i$. But in that branch, $e \Dlt a_i$.
Thus, it must be that $\hat s(a_i) \neq 0$.
The same argument applies to primed terms $t'_i$, meaning $s_X$ is well defined.

In addition, $s_X$ agrees with $s$ on all variables in $\Delta$, thus $s_X\models \Delta \land \Theta_X$.
A similar argument constructs another stack $s$ such that $s\models\Delta\land\Theta_X\land\Theta_Y$.

We now define several heaps.
\[\begin{array}{r@{\;\defeq\;}lr}
\Aptos_i & \heaplet{s(t_i)}{s(u_i)} & \forall i\in[1,k] \\[10pt]
\Aarrays_i & \bigsepcirc_{\nu=s(a_i)}^{s(b_i)} \heaplet{\nu}{0} & \forall i\in[1,n]\\[10pt]
\mathcal{A} & \bigsepcirc_{i=1}^k \Aptos_i \circ \bigsepcirc_{i=1}^n \Aarrays_i &
\end{array}\]
It is easy to see that, as $s\models\Delta$ and $\Delta\models\gamma(A)$, all
heaps $\Aptos_i$ and $\Aarrays_i$ are well-defined and disjoint.
As a consequence, $\mathcal{A}$ is well-defined. By construction, $s,\mathcal{A}\models A$.
We continue by defining heaps $\Xptos_i$ and $\Xarrays_i$.
\[\begin{array}{@{}r@{\,\defeq\,}l@{}r@{}}
\Xptos_i &
\begin{cases}
e & s(v_i)\in\dom{\mathcal{A}}\\
\heaplet{s(v_i)}{s(w_i)} & \text{otherwise}
\end{cases} & \forall i\in[1,\ell]\\[15pt]
\Xarrays_i &
\bigsepcirc_{\nu \in [s(c_i),s(d_i)] \setminus \dom{\mathcal{A}}}\heaplet{\nu}{0}
&\forall i\in[1,m] \\[15pt]
h & \mathcal{A} \circ \bigsepcirc_{i=1}^{\ell} \Xptos_i \circ \bigsepcirc_{i=1}^m \Xarrays_i &
\end{array}\]
First, observe that, by construction, $\Xarrays_i\hash\mathcal{A}$ ($i\in[1,m]$)
and $\Xptos_i\hash\mathcal{A}$ ($i\in[1,\ell]$).

Also, note that $\Xptos_i\hash\Xptos_j$ for $i\neq j \in[1,\ell]$ since otherwise
$s(v_i)=s(v_j)$ which contradicts $s\models\gamma(B)$, deriving from $s\models\Delta$
and $\Delta\models\gamma(B)$. Equally, $\Xarrays_i\hash\Xarrays_j$ for $i\neq j \in[1,m]$
by a similar argument. Finally, $\Xptos_i\hash\Xarrays_j$ for $i\in[1,\ell]$ and $j\in[1,m]$
as otherwise $s(c_j) \le s(v_i) \le s(d_j)$, contradicting again $s\models\gamma(B)$.
Thus, $h$ is well-defined.

It is not hard to verify that for each $i\in[1,\ell]$,
\[s,\Xptos_i\models\cutheap{\ptocov{B,\Delta}{v_i}{w_i}}\ .\]

We show the last obligation, i.e., that for $i\in[1,m]$
\[s,\Xarrays_i\models\cutheap{\arrcov{B,\Delta}{c_i}{d_i}}\ .\]
Suppose the opposite. Due to the form of the result returned by $\arrcov{B,\Delta}{c_i}{d_i}$
as guaranteed by \MakeUppercase prop.\nobreakspace \ref {prop:termination}, this means there must exist some address $\nu$ such that either
\mbox{$\nu \in\dom{\Xarrays_i}\setminus\dom{\cutheap{\arrcov{B,\Delta}{c_i}{d_i}}}$},
or conversely, $\nu \in\dom{\cutheap{\arrcov{B,\Delta}{c_i}{d_i}}}\setminus\dom{\Xarrays_i}$.

In the first case, there must be some $\absarray{e}{f}$ returned by $\arrcov{B,\Delta}{c_i}{d_i}$
such that $[s(e),s(f)]\not\subseteq\dom{\Xarrays_i}$. We know, however, that
$[s(e),s(f)]\subseteq[s(c_i),s(d_i)]$, from  \MakeUppercase prop.\nobreakspace \ref {prop:termination}. But then
$\nu\in[s(c_i),s(d_i)]$ thus, by the definition of $\Xarrays_i$, it must be that
$\nu\in\dom{\mathcal{A}}$. This contradicts both of the cases where an array is
returned by $\arrcovalg$ (lines\nobreakspace \ref {ln:array-t-u} and\nobreakspace  \ref {ln:array-t-ai}).

In the second case, there is some address $\nu\in\dom{\Xarrays_i}$ such that there is
no $\absarray{e}{f}$ returned by $\arrcov{B,\Delta}{c_i}{d_i}$, such that
$\nu\in[s(e),s(f)]$. Again, by assumption we have $\nu\in[s(c_i),s(d_i)]$.
However, it can be verified by inspecting $\arrcovalg$ that
if $\nu\in[s(c_i),s(d_i)]$ and there is no $\absarray{e}{f}$ in the result, this
is because $\nu\in\dom{\mathcal{A}}$, contradicting the fact that $\Xarrays_i\hash\mathcal{A}$.
This completes the proof.
\end{proof}

\begin{applemma}{\ref{lem:seqs-well-defined}}
All elements of $\Barrays,\Bptos,\Yarrays,\Yptos$ are well-defined,
in the sense that there exist such (unique) heaps.
\end{applemma}
\begin{proof}
Uniqueness follows by the above observation that all quantifier-free formulas
in $\ASL$ are precise. Here we show existence.

Suppose $\Bptos_i$ is not well-defined, meaning $s(v_i)\notin\dom{h}$, or that
$h(s(v_i))\neq s(w_i)$.
In the first case, it must be that $\ptocov{A,\Delta}{v_i}{w_i} = \emp$
(otherwise, by construction, $s(v_i)\in\dom{\cutheap{X}}$).
But this happens exactly when $s(v_i)\in\dom{\cutheap{A}}\subseteq\dom{h}$, contradiction.
In the second case, suppose $s(v_i)\notin\dom{\cutheap{A}}$.
This means $\ptocov{A,\Delta}{v_i}{w_i} = \psto{v_i}{w_i}$ which by construction
guarantees $h(s(v_i))=s(w_i)$. Finally, suppose $s(v_i)\in\dom{\cutheap{A}}$.
Either there is an $\absarray{a_j}{b_j}$ such that $s(v_i)\in[s(a_i),s(b_i)]$,
or there is $\psto{t_j}{w_j}$ such that $s(v_i)=s(t_j)$. The first possibility
contradicts the second conjuct of $\beta(A,B)$ and the second possibility the
third conjunct.

Suppose $\Barrays_i$ is not well-defined, meaning $[s(c_i),s(d_i)]\not\subseteq\dom{h}$.
In other words, there is $\nu\in[s(c_i),s(d_i)]$, but $\nu\notin\dom{h}$.
Clearly, $\nu\notin\dom{\cutheap{A}}$. By inspecting $\arrcovalg$, however,
we can conclude that there must be some $\absarray{e}{f}$ returned by
$\arrcov{A,\Delta}{c_i}{d_i}$ such that $\nu\in[s(e),s(f)]$. This means $\nu\in\dom{\cutheap{X}}$,
contradiction.

Suppose $\Yptos_i$ is not-well defined. This must mean $\Yptos_i\neq e$, because
trivially $e\subseteq h$. For this to happen, $\ptocov{B,\Delta}{t_i}{u_i}$
must return $\psto{t_i}{u_i}$, and $s(t_i)\notin\dom{h}$. But by assumption,
$s,h\models A*X$, therefore $s(t_i)\in\dom{h}$, contradiction.

Suppose $\Yarrays_i$ is not well-defined. Thus, there is some $\absarray{e}{f}$
returned by $\arrcov{B,\Delta}{a_i}{b_i}$ such that $[s(e),s(f)]\not\subseteq \dom{h}$.
However, we know $s(a_i)\leq s(e) \leq s(f) \leq s(b_i)$ from \MakeUppercase prop.\nobreakspace \ref {prop:termination}.
Also, by assumption, $s,h\models A*X$ thus $[s(a_i),s(b_i)]\subseteq\dom{h}$, contradiction.
\end{proof}

\begin{applemma}{\ref{lem:disjointness}}
\begin{enumerate}
\item For any sequence of heaps $\mathcal{S}$ of $\Barrays$,$\Bptos$,$\Yarrays$,$\Yptos$,
and any distinct $i,j\in[1,|\mathcal{S}|]$, $\mathcal{S}_i\hash\mathcal{S}_j$.
\item For any two distinct sequences of heaps $\mathcal{S},\mathcal{T}$ of $\Barrays$,
$\Bptos$, $\Yarrays$, $\Yptos$,
and any $i\in[1,|\mathcal{S}|]$, $j\in[1,|\mathcal{T}|]$, $\mathcal{S}_i\hash\mathcal{T}_j$.
\end{enumerate}
\end{applemma}
\begin{proof}
For $\Barrays$ and $\Bptos$, this follows from the fact that $s\models\gamma(B)$, ensuring
the separation of arrays and $\mapsto$ formulas in $B$.

We show that for any $i\neq j$, $\Yptos_i$ and $\Yptos_j$ are disjoint. By inspecting
$\ptocov{B,\Delta}{t_i}{u_i}$, we see that $\Yptos_i$ is either $\emptyset$,
or $\{s(t_i)\}$. If either of $\Yptos_i,\Yptos_j$ is the empty heap $e$, for $i\neq j \in[1,k]$,
then clearly $\Yptos_i\hash\Yptos_j$. If both are non-empty, then
their domains are $\{s(t_i)\},\{s(t_j)\}$ (line\nobreakspace \ref {ln:pto-notcovered}). But, by assumption, $s\models\gamma(A)$
which guarantees $s(t_i)\neq s(t_j)$.

For any $i\neq j$, $\Yarrays_i$ and $\Yarrays_j$ are disjoint, because item~(3) of \MakeUppercase prop.\nobreakspace \ref {prop:termination}
means that $\Yarrays_i\subseteq\cutheap{\absarray{a_i}{b_i}}$, for $i\in[1,n]$.
But $\cutheap{\absarray{a_i}{b_i}}\hash\cutheap{\absarray{a_j}{b_j}}$ for $i\neq j$,
due to $s\models\gamma(A)$.

We also need to show that for any pair of heaps from any two of these sequences, the heaps are disjoint.
In the case of heaps $\Barrays_i, \Bptos_j$, $\Barrays_i\hash \Bptos_j$
follows again from the assumption that $s\models\gamma(B)$.

Suppose it is not the case that $\Bptos_i\hash \Yptos_j$.
As argued previously, it must be that $\dom{\Yptos_j} = \{s(t_j)\}$.
At the same time, $\dom{\Bptos_i} = \{s(v_i)\}$, meaning that \mbox{$s(t_j)=s(v_i)$}.
Since, $t_j,v_i\in\abdterms{A}{B}$, it must be that $t_j\Deq v_i$. But then,
$\ptocov{B,\Delta}{t_i}{u_i}$ would return $\emp$ (line\nobreakspace \ref {ln:pto-pto}), contradiction.

Suppose it is not the case that $\Barrays_i\hash \Yptos_j$.
Again, this means $\dom{\Yptos_j}=\{s(t_j)\}$.
Since $\dom{\Barrays_i} = [s(c_i), s(d_i)]$,
we conclude that \mbox{$c_i \Dleq t_j \Dleq d_i$}.
We again have a contradiction, as in this case
$\ptocov{B,\Delta}{t_i}{u_i}$ would return $\emp$ (line\nobreakspace \ref {ln:pto-arr}).

Next, suppose it does not hold that $\Yarrays_i\hash\Yptos_j$.
As above, this means $\dom{\Yptos_j}=\{s(t_j)\}$.
In addition, there must be some $\absarray{e}{f}$ returned by
$\arrcov{B,\Delta}{a_i}{b_i}$ such that $\cutheap{\absarray{e}{f}}\hash\Yptos_j$ does not hold,
meaning that $\hat\Delta\models e \leq t_j \leq f$. By \MakeUppercase prop.\nobreakspace \ref {prop:termination}
we know that $\hat\Delta\models a_i \leq e \leq f \leq b_i$ thus $a_i\Dleq t_j\Dleq b_i$.
This contradicts the assumption $s\models\gamma(A)$.

Now suppose it is not the case that $\Bptos_i\hash\Yarrays_j$.
Note that $\dom{\Bptos_i} = \{s(v_i)\}$. As above,
there must be some $\absarray{e}{f}$ in the result of
$\arrcov{B,\Delta}{a_j}{b_j}$ such that $s(e)\le s(v_i)\le s(f)$,
thus $a_j\Dleq v_i\Dleq b_j$, contradicting the
second conjunct of \MakeUppercase defn.\nobreakspace \ref {defn:beta}.

Finally, we need to show that $\Barrays_i,\Yarrays_j$ are disjoint.
Suppose the contrary. This means that there is an array
$\absarray{e}{f}$ in the result of $\arrcov{B,\Delta}{a_j}{b_j}$
such that $c_i\Dleq e\Dleq d_i$ or $c_i\Dleq f\Dleq d_i$.
We inspect the return statements of $\arrcovalg$ where an array is
constructed.
At line\nobreakspace \ref {ln:array-t-u}, the array constructed is $\absarray{e}{f}$.
At this point there is no $q\in[1,m+\ell]$
such that $\hat{c_q}\Dleq e\Dleq \hat{d_q}$ (because of line\nobreakspace \ref {ln:t-covered})
or $\hat{c_q}\Dleq f\Dleq \hat{d_q}$ (because of line\nobreakspace \ref {ln:e-empty}).
At line\nobreakspace \ref {ln:array-t-ai}, the array constructed is $\absarray{e}{\hat{c_q}'}$
for some $q\in[1,m+\ell]$ such that $e \Dlt \hat{c_q} \Dleq f$, and $\hat{c_q}$
is the $\Dlt$-minimal such array endpoint. Clearly, there is no $r$ such that
$\hat{c_r}\Dleq e\Dleq \hat{d_r}$ (again because of line\nobreakspace \ref {ln:t-covered}).
Thus we need only show that there is no $r$ such that $\hat{c_r}\Dleq \hat{c_q}' \Dleq \hat{d_r}$.
But this is provided directly by the fact that $s\models\gamma(B)$.
\end{proof}

\begin{applemma}{\ref{lem:totality}}
\[
\dom{h} \subseteq
\bigcup_{i=1}^m \Barrays_i \cup
\bigcup_{i=1}^\ell \Bptos_i \cup
\bigcup_{i=1}^n \Yarrays_i \cup
\bigcup_{i=1}^k \Yptos_i
\]
\end{applemma}
\begin{proof}
We show that for all atomic formulas $\sigma$ of $F_{A*X}$
there is a set of heaps $\mathcal{H}$ from the above sequences such that
$\cutheap{\sigma}\subseteq \bigcirc H$.

Recall that $F_{A*X} = F_A * F_X$ and that
\[\begin{array}{r@{\,}l}
F_A &= \displaystyle\bigsepstar^n_{i=1}\absarray{a_i}{b_i} * \bigsepstar^{k}_{i=1}\psto{t_i}{u_i}\\
F_X & = \displaystyle\bigsepstar_{i=1}^{m} \arrcov{A,\Delta}{c_i}{d_i} * \bigsepstar_{i=1}^{\ell} \ptocov{A,\Delta}{v_i}{w_i}
\end{array}\]
We deal with the four subcases depending on the provenance of $\sigma$.

Let $\sigma\equiv \psto{t_i}{u_i}$ for some $i\in[1,k]$.
The call $\ptocov{B,\Delta}{t_i}{u_i}$
will return $\psto{t_i}{u_i}$ or $\emp$.  In the former case $\cutheap{\sigma}=\Yptos_i$.
Otherwise, there is some $j\in[1,\ell]$ such that $t_i\Deq v_j$
or there is $j\in[1,m]$ such that $c_j\Dleq t_i\Dleq d_j$. In the first
case, $\cutheap{\sigma}=\Bptos_j$ and in the second, $\cutheap{\sigma}\subseteq\Barrays_j$.

Let $\sigma\equiv\psto{e}{f}$, returned by $\ptocov{A,\Delta}{v_i}{w_i}$
for $i\in[1,\ell]$. By inspecting $\ptocovalg$
it can be seen that, necessarily, $e\equiv v_i$. But then, trivially, $\cutheap{\sigma}=\Bptos_i$
and we are done.

Let $\sigma\equiv\absarray{e}{f}$, returned by $\arrcov{A,\Delta}{c_i}{d_i}$
for $i\in[1,m]$. By \MakeUppercase prop.\nobreakspace \ref {prop:termination} we know that $c_i\Dleq e \Dleq f \Dleq d_i$, meaning
that $\cutheap{\sigma}\subseteq\Barrays_i$.

Let $\sigma\equiv\absarray{a_i}{b_i}$ for some $i\in[1,n]$,
and let $\HB = \bigcirc_{i=1}^\ell \Bptos_i \circ \bigcirc_{i=1}^m \Barrays_i$.
We argue that $\cutheap{\sigma}\subseteq\Yarrays_i\circ\HB$.
We do this by proving that for any $e,f\in\abdterms{A}{B}$ such that
$a_i\Dleq e\Dleq f\Dleq b_i$,
$\cutheap{\arrcov{B,\Delta}{e}{f}}\subseteq\Yarrays_i\circ\HB$, and we do this
by induction over the recursion depth.

If the depth is zero, then there is no array in $\abstr{B}$ covering $e$ and
there is no array covering $f$ (line\nobreakspace \ref {ln:array-t-u}).
Thus, $\arrcov{B,\Delta}{e}{f}=\absarray{e}{f}$, therefore trivially
$\cutheap{\arrcov{B,\Delta}{e}{f}}\subseteq\cutheap{\absarray{a_i}{b_i}}=\Yarrays_i$.

If the depth is non-zero, either there is $\absarray{\hat{c_j}}{\hat{d_j}}$
such that $s(e) \in[s(\hat{c_j}),s(\hat{d_j})]$ (line\nobreakspace \ref {ln:t-covered})
or there is no such array, but
there is a (left-most) $\absarray{\hat{c_j}}{\hat{d_j}}$ covering $f$ (line\nobreakspace \ref {ln:i-minimal}).

In the first case, $\arrcov{B,\Delta}{e}{f} = \arrcov{B,\Delta}{\hat{d_j}+1}{f}$,
where $a_i\Dleq e\Dleq \hat{d_j} \Dlt \hat{d_j}+1$. If $f \Dlt \hat{d_j}+1$ then
the call $\arrcov{B,\Delta}{\hat{d_j}+1}{f}=\emp$ and we are done. Otherwise,
the inductive hypothesis applies and we get
$\arrcov{B,\Delta}{e}{f} = \arrcov{B,\Delta}{\hat{d_j}+1}{f}\subseteq\Yarrays_i\circ\HB$.

In the second case, $\arrcov{B,\Delta}{e}{f}$ is equal to
\[
\hat{c}'_j+1=\hat{c_j}:\absarray{e}{\hat{c}'_j} * \arrcov{B,\Delta}{\hat{d}_j+1}{f}\ .
\]
This means that $\cutheap{\arrcov{B,\Delta}{e}{f}}$ is equal to
\[
\cutheap{\absarray{e}{\hat{c}'_j}} \cup \cutheap{\arrcov{B,\Delta}{\hat{d}_j+1}{f}}
\]
Clearly, $\cutheap{\absarray{e}{\hat{c}'_j}}\subseteq\Yarrays_i$ so we need to show
the same for $\cutheap{\arrcov{B,\Delta}{\hat{d}_j+1}{f}}$. This follows by an identical
argument to the previous case, via the inductive hypothesis. This completes the proof.
\end{proof}

\begin{applemma}{\ref{lem:delta-polytime}}
Let $(A,B)$ be a biabduction instance and
$\Delta$ a formula satisfying \MakeUppercase conditions\nobreakspace \ref {item:terms} and\nobreakspace  \ref {item:order} of
\MakeUppercase defn.\nobreakspace \ref {defn:seed}. Let $\Gamma=\bigwedge\bigvee \pi$ be a formula where
$\pi$ is of the form $t< u$ or $t = u$ and $t,u\in\abdterms{A}{B}$. Then,
checking $\Delta\modelsPA \Gamma$ is in $\PTIME$.
\end{applemma}
\begin{proof}
Let $\Delta=\bigwedge_{i\in I}\delta_i$.
First, we assume that $J=K=\{1\}$, i.e., that the query in question is simply $\Delta\modelsPA\pi$
for a single atomic formula $\pi$. Let $\pi$ be of the form $t < u$.
If there exists $i\in I$ such that $\delta_i \equiv \pi$ then
clearly $\Delta\modelsPA\pi$, as $\Delta$ is a conjunction. We return ``yes''.

If there is no $i\in I$ such that $\delta_i\equiv\pi$ then by the assumption
that $t,u\in\abdterms{A}{B}$ and \MakeUppercase condition\nobreakspace \ref {item:order} of \MakeUppercase defn.\nobreakspace \ref {defn:seed} we have
that $t = u$ or $u < t$ is a conjunct of $\Delta$. In both cases it is clear
that $\Delta\not\modelsPA\pi$ (again because $\Delta$ is a conjunction)
and we return ``no''.

The case where $\pi$ is of the form $t = u$ is almost identical. Thus we
can answer queries of the form  $\Delta\modelsPA\pi$ in time linear in $|\Delta|$.

Suppose now that $\Gamma = \bigvee_{k\in K} \pi_k$, that is, $I$ is a singleton.
We issue all possible queries of the form $\Delta\modelsPA\pi_k$ for $k\in K$.
If any of these queries reports ``yes'' then we report ``yes''.  Otherwise,
due to the completeness of checking these queries we have that
$\Delta\modelsPA \bigwedge_{k\in K} \lnot \pi_k$ and we report ``no''.
Therefore, queries of the form $\Delta\modelsPA\bigvee_{k\in K} \pi_k$ can be checked
in time $|\Delta|\cdot |K|$.

Finally, suppose $\Gamma=\bigwedge_{j\in J}\bigvee_{k\in K} \pi_{j,k}$.
We issue $|J|$ queries of the form $\Delta\modelsPA\bigvee_{k\in K} \pi_{j,k}$
for each $j\in J$. If all queries receive positive answers then clearly
$\Delta\modelsPA\Gamma$ and we return ``yes''. Otherwise there is $j\in J$ such
that the query $\Delta\modelsPA\bigvee_{k\in K} \pi_{j,k}$ received a negative answer,
meaning that $\Delta\not\modelsPA\bigvee_{k\in K} \pi_{j,k}$. Thus, as $\Gamma$
is a conjunction at the top-level, it is clear that $\Delta\not\modelsPA\Gamma$
and we report ``no''. This last step can take up to $|\Delta|\cdot|K|\cdot|J|$ time.
\end{proof}

\begin{appprop}{\ref{prop:runtime_NP}}
Deciding if there is a solution for a biabduction problem $(A,B)$,
and constructing it if it exists, can be done in $\NP$.
\end{appprop}

\begin{proof}
We outline an non-deterministic algorithm that runs in polynomial time in the
size of the input $(A,B)$.

First, we guess a set $T\subseteq\abdterms{A}{B} \times \{ <, = \} \times \abdterms{A}{B}$.
Next, we guess an assignment of values $s$ to the variables in $\abdterms{A}{B}$.
We limit the range of the assignment to naturals bounded by some $B$
which is exponential in $|\abdterms{A}{B}|$
(representable in polynomial space and guessable in non-deterministic polynomial time).
The precise definition of $B$ is not relevant here, and is given in
\cite[Theorem~6]{Scarpellini:84}.

We then convert the set $T$ (of quadratic size in $|\abdterms{A}{B}|$)
into a formula $\Delta$ in the obvious way.
The resulting $\Delta$ automatically satisfies \MakeUppercase condition\nobreakspace \ref {item:terms} of \MakeUppercase defn.\nobreakspace \ref {defn:seed}.
\MakeUppercase condition\nobreakspace \ref {item:order} of \MakeUppercase defn.\nobreakspace \ref {defn:seed} is checkable in quadratic time by a
nested loop over pairs of terms from $\abdterms{A}{B}$, scanning $\Delta$ in each iteration.
The formula $\beta(A,B)$ can be split into a fixed number of formulas
of the form $\bigwedge\bigvee \pi$, given that $t \le u$ is equivalent to \mbox{$t<u \lor t=u$}.
Thus, $\Delta\models\beta(A,B)$ can be checked in polynomial
time due to Lemma\nobreakspace \ref {lem:delta-polytime}. Finally, we check that $\Delta$
is satisfiable by checking whether $s\modelsPA \Delta$.
This step can be done in polynomial time and is complete
by \cite[Theorem~6]{Scarpellini:84}. If all checks pass, then $\Delta$ is a solution seed.

We now apply \MakeUppercase defn.\nobreakspace \ref {alg:biabduction} on $\Delta$ and obtain
the formulas $X$ and $Y$. By \MakeUppercase prop.\nobreakspace \ref {prop:termination}, each call
$\arrcov{A,\Delta}{c_j}{d_j}$ issues at most $n+k$ recursive calls.
The work done in each call is clearly doable in polynomial time
(cf.~Lemma\nobreakspace \ref {lem:delta-polytime}), thus completing the proof.
\end{proof}

\begin{appprop}{\ref{prop:biabd_qf}}
Let $A$ be quantifier-free, and let $B$ be such that no variable appearing in the RHS of a $\mapsto$ formula is existentially bound.
Then an instance $(A,B)$ of the biabduction problem for $\ASL$ has a solution
if and only if $(A,\qf{B})$ has a solution.
\end{appprop}

\begin{proof}
Let $B = \exists\vec{z}.\ Q$, where $Q=\qf{B}$ is quantifier-free. We tackle each direction of the equivalence separately. \medskip

\noindent ($\Leftarrow$) Let $(X,Y)$ be a solution for $(A,Q)$. We claim that $(X,Y)$ is also a solution for $(A,\exists\vec{z}. Q)$. To see this, observe that by assumption $A*X$ is satisfiable and $A * X \models Q * Y$. Since trivially $Q \models \exists\vec{z}.\ Q$, we easily have $Q * Y \models (\exists \vec{z}.\ Q) * Y$ and so $A*X \models (\exists \vec{z}.\ Q) * Y$ as required. \medskip

\noindent($\Rightarrow$) Let $(X,Y)$ be a solution for $(A, \exists\vec{z}.\ Q)$. That is, $A * X$ is satisfiable and $A * X \models (\exists\vec{z}.\ Q) * Y$. Since the free variables in $Y$ are disjoint from $\vec{z}$, this can be rewritten as $A * X \models \exists\vec{z}.\ (Q * Y)$.  Now, by assumption there is a stack-heap pair $(s,h)$ such that $s,h \models A*X$. Furthermore, $h$ is clearly independent of the data values stored in the arrays in $A * X$.  Thus we may choose $h$ such that $h(x) \neq s(w)$ for all formulas of the form $\psto{v}{w}$ occurring in $Q * Y$, and for all $x$ such that $\absarray{a}{b}$ occurs in $A * X$ and $s(a) \leq s(x) \leq s(b)$.

Now, since $A * X \models (\exists\vec{z}.\ Q) * Y$, we get $s,h \models \exists\vec{z}.\ (Q * Y)$, meaning that $s[\vec{z} \mapsto \vec{m}], h \models Q * Y$ for some $\vec{m}$.
We write $s' = s[\vec{z}\mapsto \vec{m}]$, and define the following extension of the symbolic heap $X$:
\[
X' \defeq \big(\textstyle\bigwedge_{x \in \fv{A,X,Q,Y}} x =s'(x) \big) * X
\]
We claim that $(X',Y)$ is then a solution for $(A,Q)$. First, since $s,h \models A*X$ but the variables $\vec{z}$ do not occur in $A * X$ by assumption, we also have $s', h \models A * X$. Clearly, we also have $s' \models x = s'(x)$ for any $x$, and so $s',h \models A * X'$. Hence $A * X'$ is satisfiable.

It remains to show that $A * X' \models Q * Y$.  Supposing $s'',h' \models A * X'$, we require to prove $s'',h' \models Q * Y$.
By construction of $X'$, the stack $s''$ agrees with $s'$ on all variables occurring in $A$, $X'$, $Q$ and $Y$, so in fact we have $s',h' \models A * X'$ and require to prove $s',h' \models Q * Y$.

Now, since $s',h \models A * X'$ and $s',h' \models A * X'$, we have $\dom{h'} = \dom{h}$ by Lemma~\ref{lem:semantics}.  Since $s',h \models Q * Y$, it is then easy to see that $s',h'$ satisfies all pure formulas and all $\textsf{array}$ formulas appearing in $Q * Y$.  The only difficulty is that $s',h'$ may fail to satisfy some formula of the form $\psto{v}{w}$ in $Q * Y$, because $h'(s'(v)) \neq s'(w)$.  Suppose for contradiction this is the case.

Since $\dom{h'} = \dom{h}$, we must have $s'(v) \in \dom{h}$, and since $s',h \models A * X$, it must be that $s'(v)$ is covered by some formula in $A * X$. If there is a formula of the form $\psto{t}{u}$ in $A * X$ such that $s'(t) = s'(v)$, then we have $s'(w) = h(s'(v)) = h(s'(t)) = s'(u)$, since $s',h \models A * X$ and $s',h \models Q * Y$.  But then $h'(s'(v)) = h'(s'(t)) = s'(u) = s'(w)$, since $s',h' \models \psto{t}{u}$, a contradiction.  Therefore, there must be a formula $\absarray{a}{b}$ in $A * X$ such that $s'(a) \leq s'(v) \leq s'(b)$.  But then, due to our initial choice of $h$, we know that $h(s'(v)) \neq s(w)$.  Since the existential variables $\vec{z}$ are not allowed to include $w$, this means $h(s'(v)) \neq s'(w)$, contradicting the fact that $s',h \models Q * Y$ (since it does not satisfy $\psto{v}{w}$).  This completes the proof.
\end{proof}

\begin{applemma}{\ref{l-colour-biabduct}}
 Given an instance $G$ of the
 \mbox{$2$-}round \mbox{$3$-}colouring problem,
 the following statements are pairwise equivalent:
\begin{itemize}
\item[(a)]
 The biabduction problem \mbox{$(A_G,B_G)$} has a solution.
\item[(b)]
 There is a winning strategy for the perfect colouring~$G$.
\item[(c)]
 \mbox{$A_G \models B_G$} is valid.    
\end{itemize}
 where $A_G$ and $B_G$ are the symbolic heaps given
 by Definition~\ref{d-colour-biabduct}  
\end{applemma}

\begin{proof}
 We establish each direction of the above equivalences in turn.

\noindent \mbox{$(c)\Rightarrow (a)$}\ \

\noindent
 This direction is trivial by taking
  $$ X=Y=\emp $$%

\noindent \mbox{$(b)\Rightarrow (c)$}\ \

 Suppose that there is a winning strategy such
 that every \mbox{$3$-}colouring of the leaves can be extended to a
 perfect \mbox{$3$-}colouring of the whole~$G$. We will prove that
\mbox{$A_{G}\models B_{G}$}.

\noindent
 Let \mbox{$s,h$} be a stack-heap pair satisfying
 \mbox{$s,h \models A_{G}$}.

\noindent
 The spatial part of $A_{G}$ yields a decomposition of~$h$ as
\begin{equation}%
 h = \bigsepcirc_{i=1}^{k} h_{i}\ \circ\
\hspace*{-1.5ex} \bigsepcirc_{(v_i, v_j)\in E}
\hspace*{-1.5ex} \widetilde{h_{ij}^{(e)}}
              \label{eq-h-A}
\end{equation}%
 where
 for each \mbox{$1\leq i\leq k$},
 the $h_i$ is a one-cell array, and for some~$b_i$,
\begin{equation}
 \dom{h_i} = \{s(d_i)\}, \ \ \mbox{and}\ \
  h_i(s(d_i)) = b_i                  
              \label{eq-h-i=k}
\end{equation}
 and, for each \mbox{$(v_i,v_j)\in E$},
\begin{equation}%
 \mbox{$\dom{\widetilde{h_{ij}^{(e)}}} =
     \{s(e_{ij})+1,\ s(e_{ij})+2,\ s(e_{ij})+3\}$}   
              \label{eq-h-ij-neq}
\end{equation}%

 Take the \mbox{$3$-}colouring of the leaves obtained by assigning
 the colours \mbox{$b_{i,1}$} to the leaves $v_1$, $v_2$,\dots, $v_k$
 resp..
 where \mbox{$1\leq b_{i,1} \leq 3$},\ \
 and \ \ \mbox{$b_{i,1}-1 \equiv b_i\ (\bmod {3})$}\ .

 According to the winning strategy, we can assign colours, denote
 them by $b_{i,1}$, \mbox{$i>k$}, to the rest of vertices\
 $v_{k+1}$, \dots, $v_n$, resp.,
 obtaining a \mbox{$3$-}colouring of the whole~$G$ such that
 no adjacent vertices share the same colour.    

 In addition, we mark edges \mbox{$(v_i,v_j)$}
 by $\widetilde{b}_{ij}$ complementary to $b_{i,1}$ and~$b_{j,1}$.

\noindent
 We extend the stack $s$ for quantified variables in~$B_{G}$
 so that for all \mbox{$i\leq k$},
  $$ s(c_{i,1}) = b_{i,1} = h_i(s(d_i)),$$
 and, for each \mbox{$(v_i, v_j)\in E$},  
$$ s(\widetilde{c_{ij}}) =  6 - b_{i,1} - b_{j,1}.$$%

 The fact that no adjacent vertices $v_i$ and\/ $v_j$
 share the same colour provides that\/\ \
 $$\mbox{$(s(c_{i,1}),\, s(c_{j,1}),\, s(\widetilde{c_{ij}}))$}$$%
 {\em is a permutation of\/}
 \ $$\mbox{$(1,\, 2,\, 3)$},$$%
 resulting in that
 \ \ \mbox{$s,\widetilde{h_{ij}^{(e)}}$}\ from~(\ref{eq-h-ij-neq})\ \
 is also a model for
 \[ \mbox{$
    \relarray{e_{ij}}{c_{i,1}}{c_{i,1}}
  * \relarray{e_{ij}}{c_{j,1}}{c_{j,1}}
  * \relarray{e_{ij}}{\widetilde{c_{ij}}}{\widetilde{c_{ij}}}$} \]%

 Bringing all together, we get that \mbox{$s,h$}
 satisfies \mbox{$s,h \models B_{G}$},
 which completes the proof of this direction.


~\\
 \mbox{$(a)\Rightarrow (b)$}

 Let \mbox{$A_{G}*X\models B_{G}*Y$} and \mbox{$A_{G}*X$}
 be satisfiable.

 Since \mbox{$A_{G}*X$} is satisfiable, there is a model
 of the form \mbox{$s, h_A\circ h_X$} such that
 $$ s, h_A \models A_{G}, \ \ \mbox{and}\ \  s, h_X \models X,$$%
 and, in particular,
 $$ h_A = \bigsepcirc_{i=1}^{k} h_{i}\ \circ\
\hspace*{-1.5ex} \bigsepcirc_{(v_i, v_j)\in E}
\hspace*{-1.5ex} \widetilde{h_{ij}^{(e)}},$$%
 where
 for each \mbox{$1\leq i\leq k$},
 the $h_i$ is a one-cell array such that
\begin{equation}
 \dom{h_i} = \{s(d_i)\}, \ \ \mbox{and}\ \
  h_i(s(d_i)) = s(c_{i})                  
              \label{eq-h-i=k-2}
\end{equation}
 and, for each \mbox{$(v_i,v_j)\in E$},
\begin{equation}%
 \mbox{$\dom{\widetilde{h_{ij}^{(e)}}} =
     \{s(e_{ij})+1,\ s(e_{ij})+2,\ s(e_{ij})+3\}$}   
              \label{eq-h-ij-neq-2}
\end{equation}%

 We will construct the required winning strategy in the following way.

 Assume a \mbox{$3$-}colouring of the leaves be given by assigning
 colours, say $b_{i,1}$, to the leaves $v_1$, $v_2$,\dots, $v_k$
 respectively.

 Then we modify our original stack~$s$ to a stack~$s'$
 by setting, for each \mbox{$1\leq i\leq k$},
 $$ s'(c_i) = b_{i,1}.$$%
 with modifying thereby $h_A$ to $h_A'$ by means of
 replacing each $h_i$ with the updated $h_i'$ in which
  $$ h_i'(s(d_i)) = s'(c_{i}) =  b_{i,1}.$$
 We claim that still
 $$ s', h_A' \models A_{G}, \ \ \mbox{and}\ \  s', h_X \models X,$$%
 and, therefore,
 $$ s', h_A'\circ h_X \models A_{G}* X.$$%
\begin{quote}
 {\em The crucial point is that
\begin{itemize}
\item[(a)]
 First, $c_i$ is quantified so that $X$ cannot
 refer to\/~$c_i$ explicitly.
\item[(b)]
 The only indirect possibility for\/~$X$
 to refer to\/~$c_i$ by applying $h_X$ to\/~$d_i$
 is blocked by the fact that $d_i$ is not in the domain of\/~$h_X$.
\end{itemize}
}
\end{quote}

\noindent
 Since \mbox{$A_{G}*X\models B_{G}*Y$}, we get
 $$ s', h_A'\circ h_X \models B_{G}*Y,$$%
 and for some\/\ \mbox{$h_B \subseteq h_A'\circ h_X$}\
\ and stack\/~$s_B$, which is extension of\/~$s'$
 to the existentially quantified variables in\/~$B$,
 $$ s_B, h_B \models B.$$%

\noindent
 Recall that $B_G$ has been defined as follows:
\[\begin{array}{@{}l} \exists \vec{z}.\ \big(
\textstyle\bigwedge_{i=1}^{n}(1\leq{c_{i,1}}\leq 3) \wedge
\textstyle\bigwedge_{(v_i,v_j)\in E}(1\leq \widetilde{c_{ij}}\leq 3)
\\[1ex]
 \wedge\ \
\textstyle\bigwedge_{i=1}^{k}\,(c_{i,1}-1 \equiv c_i\ (\bmod {3}))
\colon \\[1ex] 
 \textstyle\bigsepstar_{i=1}^{k} d_i \mapsto c_i \ *
 \textstyle\bigsepstar_{(v_i,v_j)\in E} 
 \relarray{e_{ij}}{c_{i,1}}{c_{i,1}}
\\[2ex]
 * \textstyle\bigsepstar_{(v_i,v_j)\in E}\,
   \relarray{e_{ij}}{c_{j,1}}{c_{j,1}}
   *\relarray{e_{ij}}{\widetilde{c_{ij}}}{\widetilde{c_{ij}}} \big).
\end{array}\]%
 where the existentially quantified variables $\vec{z}$
 are all variables occurring in\/~$B_{G}$ that are not mentioned
 explicitly in\/~$A_{G}$. 

\noindent
 Because of \mbox{$d_i \mapsto c_i$},
 we have for each \mbox{$1\leq i\leq k$},
 $$ s_B(c_i) =  h_i'(s(d_i)) = s'(c_{i}) = b_{i,1},$$%
 which means that, for \mbox{$1\leq i\leq k$},
 these \mbox{$s_B(c_i)$} represent correctly
 the original \mbox{$3$-}colouring of the leaves.

 Take the \mbox{$3$-}colouring of the whole~$G$ obtained
 by assigning the colours \mbox{$s_B(c_{i,1})$}
 to the rest of vertices $v_{k+1}$,\dots,$v_n$ respectively.

\noindent
 The part of the form
$$  \relarray{e_{ij}}{c_{i,1}}{c_{i,1}}
  *
  \relarray{e_{ij}}{c_{j,1}}{c_{j,1}}
  *\relarray{e_{ij}}{\widetilde{c_{ij}}}{\widetilde{c_{ij}}},$$
 provides that\ \mbox{$s_B(c_{i,1})\neq s_B(c_{j,1})$},\
 which results in that
 no adjacent vertices $v_i$ and\/ $v_j$ share
 the same colours \mbox{$s_B(c_{i,1})$} and\/ \mbox{$s_B(c_{j,1})$},
 with providing a perfect \mbox{$3$-}colouring of\/~$G$.

\noindent
 This completes the direction, and
 the proof.
\end{proof}

\section{Proofs of results in Section~\ref{sec:entailment}}

In order to prove Lemma~\ref{lem:entail_coding}, we make use of the following simple auxiliary lemma about the formula $\phi(-,-)$ from Definition~\ref{defn:entail_phi}.

\begin{lemma}
\label{lem:entail_phi}
Let $A$ and $B$ be symbolic heaps with respective spatial parts:
\[\begin{array}{rl}
A: & \bigsepstar^n_{i=1}\absarray{a_i}{b_i} * \bigsepstar^k_{i=1}\psto{t_i}{u_i} \\
B: & \bigsepstar^m_{j=1}\absarray{c_j}{d_j} * \bigsepstar^\ell_{j=1}\psto{v_j}{w_j}
\end{array}\]
Then we have, for any stack $s$,
\[\begin{array}{c@{\hspace{0.2cm}}l}
& s \models \phi(A,B) \\
\iff & \exists y \in \bigcup^n_{i=1}\{s(a_i),\ldots,s(b_i\} \cup \bigcup_{i=1}^k\{s(t_i)\} \mbox{ such that}\\
& \quad y \notin \bigcup_{j=1}^m\{s(c_j),\ldots,s(d_j)\} \cup \bigcup_{j=1}^\ell\{s(v_j)\}\ .
\end{array}\]
where $\phi(-,-)$ is given by Defn.~\ref{defn:entail_phi}. 
\end{lemma}

\begin{proof}
Follows straightforwardly from the definitions of $\phi(-,-)$ and $\abstr{-}$.
\end{proof}

\begin{applemma}{\ref{lem:entail_coding}}
For any instance $(A,B)$ of the $\ASL$ entailment problem, and for any stack $s$,
\[
s \modelsPA \chi(A,B) \;\iff\; \exists h.\ s,h \models A\ \mbox{ and } s,h\not\models B\ .
\]%
\end{applemma}

\begin{proof}
We assume that $A$ and $B$ are of the form given by Lemma~\ref{lem:entail_phi}, and establish each direction of the lemma separately.

\medskip\noindent($\Rightarrow$) Supposing that $s \modelsPA \chi(A,B)$, we require to construct a heap $h$ such that $s,h \models A$ and $s,h \not\models B$. Note that $s \modelsPA \gamma(A)$ by assumption, so by Lemma~\ref{lem:gamma} there is a heap $h$ such that $s,h \models A$.  Moreover, this fact is clearly independent of the data values stored in the arrays in $A$.  Thus we may choose $h$ such that $h(x) \neq s(w_j)$ for all $j \in[1,\ell]$ and for all $x \in [s(a_i),s(b_i)]$, where $i \in [1,n]$.

Now suppose for contradiction that $s,h \models B$. Thus, for some $\vec{q} \in \val^{|\vec{z}|}$ we have $s[\vec{z}\mapsto\vec{q}],h \models \qf{B}$ (where $\vec{z}$ is the tuple of existentially quantified variables in $B$). For convenience, we write $s' \defeq s[\vec{z}\mapsto\vec{q}]$. Since $\vec{z}$ does not include any variable in $A$, we also have $s',h \models A$.  Thus
\[
\dom{h} = \bigcup^n_{i=1}\{s'(a_i),\ldots,s'(b_i\} \cup \bigcup_{i=1}^k\{s'(t_i)\}
\]
and $h(s'(t_i)) = s'(u_i)$ for all $i \in [1,k]$.  Similarly, since $s',h \models \qf{B}$, we have
\[
\dom{h} = \bigcup_{j=1}^m\{s'(c_j),\ldots,s'(d_j)\} \cup \bigcup_{j=1}^\ell\{s'(v_j)\}
\]
and $h(s'(v_j)) = s'(w_j)$ for all $j \in [1,\ell]$.  We note that, because of our restrictions on existential quantification, $s'(w_j) = s(w_j)$ for all $w_j$.

Now, since $s \models \chi(A,B)$, by instantiating the universal quantifiers $\forall\vec{z}$ in the second conjunct as $\vec{q}$, we obtain
\[
s' \models \neg\gamma(\qf{B}) \vee \phi(A,B) \vee \phi(B,A) \vee \psi_1(A,B) \vee \psi_2(A,B).
\]
However, since $s',h \models \qf{B}$, we have $s' \models \gamma(\qf{B})$ by Lemma~\ref{lem:gamma}, and therefore
\[
s' \models \phi(A,B) \vee \phi(B,A) \vee \psi_1(A,B) \vee \psi_2(A,B).
\]
This gives us four disjunctive subcases to consider.

\proofcase{$s' \models \phi(A,B)$} In this case, Lemma~\ref{lem:entail_phi} and the two equations above for $\dom{h}$ imply that there exists $y \in \dom{h}$ such that $y \not\in \dom{h}$; contradiction.

\proofcase{$s' \models \phi(B,A)$} Symmetric to the case above.

\proofcase{$s' \models \psi_1(A,B)$} We have $s'(a_i) \leq s'(v_j) \leq s'(b_i)$ for some $i \in [1,n]$ and $j \in [1,\ell]$.  On the one hand, we have $h(s'(v_j)) = s'(w_j) = s(w_j)$.  On the other hand, $h$ was chosen specifically such that $h(x) \neq s(w_j)$ for any $x \in [s(a_i),s(b_i)] (=[s'(a_i),s'(b_i)])$.  Hence we have a contradiction.

\proofcase{$s' \models \psi_2(A,B)$} We have $s'(t_i) = s'(v_j)$ and $s'(u_i) \neq s'(w_j)$ for some $i \in [1,k]$ and $j \in [1,\ell]$. On the one hand we have $h(s'(t_i)) = s'(u_i) \neq s'(w_j)$, and on the other we have $h(s'(t_i)) = h(s'(v_j)) = s'(w_j)$, a contradiction.  This completes all subcases.

\medskip\noindent($\Leftarrow$) Supposing that $s,h \models A$ but $s,h\not\models B$, we need to show that $s \models \chi(A,B)$.  Since $s,h \models A$, we immediately get $s \models \gamma(A)$ by Lemma~\ref{lem:gamma}.  Then, letting $\vec{q} \in \val^{|\vec{z}|}$ be an arbitrary instantiation of the variables $\vec{z}$ and writing  $s'=s[\vec{z}\mapsto\vec{q}]$, it remains show that
\[
s' \models \neg\gamma(\qf{B}) \vee \phi(A,B) \vee \phi(B,A) \vee \psi_1(A,B) \vee \psi_2(A,B)\ .
\]
Since $\vec{z}$ does not mention any variable in $A$, we have $s',h \models A$, and thus
\[
\dom{h} = \bigcup^n_{i=1}\{s'(a_i),\ldots,s'(b_i\} \cup \bigcup_{i=1}^k\{s'(t_i)\}
\]
with $h(s'(t_i)) = s'(u_i)$ for all $i \in [1,k]$.

Now, since $s, h \not\models B$, we can instantiating the quantifiers $\vec{z}$ in $B$ by $\vec{q}$ to obtain $s',h \not\models \qf{B}$.  If $s',h \not\models \Pi'$, then immediately $s' \models \neg\gamma(\qf{B})$ and we are done.  Otherwise, $s',h$ fails to satisfy the spatial part of $\qf{B}$. By examining the satisfaction relation for spatial formulas, this yields four disjunctive subcases.
\begin{enumerate}
\item Some array in $B$ is ill-defined under $s'$, i.e. $s'(c_j) > s'(d_j)$ for some $j$.  In that case $s' \models \neg\gamma(\qf{B})$, and we are done.

\item Each array in $B$ is defined under $s'$, but $\dom{h}$ is not, because the domains of the arrays and pointers in $B$ overlap on some location.  In this case, it is again straightforward to see that $s' \models \neg\gamma(\qf{B})$.

\item The domain $\dom{h}$ is well-defined, but not equal to
\[
\textstyle\bigcup_{j=1}^m\{s'(c_j),\ldots,s'(d_j)\} \cup \bigcup_{j=1}^\ell\{s'(v_j)\}\ .
\]
In that case, using Lemma~\ref{lem:entail_phi} and the characterisation of $\dom{h}$ in terms of $A$ above, it is easy to show that either $s' \models \phi(A,B)$ or $s' \models \phi(B,A)$.

\item Finally, it might be that $\dom{h}$ agrees with the spatial part of $B$ under $s'$ (i.e. $s',h \models \abstr{\qf{B}}$), but disagrees on some pointer value, i.e., $h(s'(v_j)) \neq s'(w_j)$ for some $j \in [1,\ell]$.  We observe that $s'(v_j) \in \dom{h}$, and distinguish two further subcases, using the previous characterisation of $\dom{h}$ in terms of $A$ above.
    \begin{itemize}
    \item If $s'(v_j) \in \{s'(a_i),\ldots,s'(b_i)\}$ for some $i \in [1,n]$, then we immediately have $s' \models \psi_1(A,B)$.
    \item Otherwise, $s'(v_j) = s'(t_i)$ for some $i \in [1,k]$.  In that case, $h(s'(v_j)) = h(s'(t_i)) = s'(u_i)$, and thus $s'(u_i) \neq s'(w_j)$.  Thus $s' \models \psi_2(A,B)$, and we are done.  This completes all subcases, and the proof.
    \end{itemize}
\end{enumerate}
\end{proof}

\begin{applemma}{\ref{l-kill-phi}}
 We can rewrite \mbox{$\Phi(A,B)$}
 as a quantifier-free formula at only polynomial cost.
\end{applemma}

\begin{proof}
We write $\Phi(A,B) = \exists x. \alpha_{A,B}(x)$, so that, following Definition~\ref{defn:entail_phi}, $\alpha_{A,B}(x)$ is the formula
\[
\textstyle\bigvee_{i=1}^{n}
     a_i \leq x \leq b_i
        \wedge \bigwedge_{j=1}^{m}(x < c_j) \vee (x > d_j)\ .
\]
We claim that $\Phi(A,B)$ is then equivalent to the formula
\[
\textstyle\bigvee_{i_0=1}^{n}
  \alpha_{A,B}(a_{i_0}) \vee \bigvee_{j_0=1}^{m} \alpha_{A,B}(d_{j_0}+1) \ .
\]

One direction of the equivalence is trivial (any stack satisfying the above formula immediately satisfies $\alpha_{A,B}(x)$ for some $x$ and therefore $\Phi(A,B)$).  We show the non-trivial direction.


Assuming that $s \models \Phi(A,B)$, there exists a number $x_0$ and $k \in [1,n]$ such that
\[
s\models a_{k} \leq x_0 \leq b_{k}
 \land \textstyle\bigwedge_{j=1}^{m}(x_0 < c_j) \vee (x_0 > d_j)\ .
\]
We consider two cases, recalling that \mbox{$\phi(A,B)$} captures the property that there is an address in an array in $A$ that is not covered by any of the arrays in $B$ (cf. Lemma~\ref{lem:entail_phi}).

\begin{enumerate}
 \item Suppose that the address~\mbox{$s(a_{k})$} is not covered
 by any array in $B$, i.e., that $s(a_k) < s(c_j)$ or $s(a_k) > s(d_j)$ for all $j \in [1,m]$.  In that case, trivially, $s \models \alpha_{A,B}(a_k)$, and we are done.
 %

\item Otherwise, $s(a_k)$ is covered by an array in $B$, i.e., $s(c_j) \leq s(a_k) \leq s(d_j)$ for some $j \in [1,m]$.
 Then we choose $d_{j_0}$ such that
\[
s(d_{j_0}) = \max_{1\leq j\leq m}\{s(d_{j})\mid  s(d_{j})< x_0 \}.
\]
(That is, $d_{j_0}$ is the largest right-endpoint of an array in $B$ that is still smaller than $x_0$.)  In that case, the effect is that \mbox{$s(d_{j_0}+1)$} must still be covered by the arrays in $A$,
\[
s\models (a_{k}\leq d_{j_0}< d_{j_0}+1\leq x_0 \leq b_{k})
\]
but \mbox{$s(d_{j_0}+1)$} cannot itself be allocated in $B$
\[\begin{array}{rl}
s\models & \textstyle\bigwedge_{j=1}^{m}((d_{j_0}+1\leq x_0 < c_j) \\
    & \mathrel{\vee} ((x_0 > d_j)\land (d_{j_0}+1 > d_j)))
\end{array}\]
Hence \mbox{$\alpha_{A,B}(d_{j_0}+1)$} holds, and we are done.\qedhere
\end{enumerate}
\end{proof}

\begin{applemma}{\ref{lem:colour_to_entail}}
Let $G$ be a 2-round \mbox{3-}colouring instance,
 and let $A_G$ and $B_G$ be the symbolic heaps given
 by Defn.~\ref{defn:colour_to_entail}. Then, we have
\[
A_G \models B_G \;\iff\;
 \exists\ \mbox{winning strategy for colouring $G$.}
\]
\end{applemma}

\begin{proof}
 Similar to Lemma~\ref{l-colour-biabduct}.
\end{proof}

\end{document}